\documentclass[10pt]{article}

\usepackage{amsmath,amsfonts,amssymb,bbm,complexity} 
\usepackage[margin=1.5in]{geometry}
\newcommand{\SNM}{\ifmmode\mathsf{SNM}\else\textsf{SNM}\fi}
\newcommand{\SLP}{\ifmmode\mathsf{SLP}\else\textsf{SLP}\fi}
\newcommand{\RSEB}{\ifmmode\mathsf{RSEB}\else\textsf{RSEB}\fi}
\newcommand{\RKOTH}{\ifmmode\mathsf{RKotH}\else\textsf{RKotH}\fi}
\newtheorem{observation}{Observation}
\newcommand{\bal}{\ensuremath{T^{\text{Bal}}}}

\newtheorem{open}{Open Question}
\newcommand{\mattnote}[1]{\textcolor{blue}{#1}}
\newcommand{\mattfootnote}[1]{\mattnote{\footnote{\mattnote{From Matt: #1}}}}
\newcommand{\arnote}[1]{\textcolor{red}{#1}}

\newtheorem{theorem}{Theorem.}[section]
\newtheorem{lemma}[theorem]{Lemma}
\newtheorem{proposition}[theorem]{Proposition}
\newtheorem{corollary}[theorem]{Corollary}
\newtheorem{claim}[theorem]{Claim}
\newtheorem{definition}{Definition}[section]
\newenvironment{proof}[1][Proof.]{\begin{trivlist}
\item[\hskip \labelsep {\bfseries #1}]}{\end{trivlist}}

\usepackage{float}
\usepackage[utf8]{inputenc}
\usepackage[english]{babel}
\usepackage{hyperref}
\usepackage{authblk}
\usepackage[ruled]{algorithm2e} % For algorithms
\usepackage[pdftex]{graphicx}

\begin{document}
% Title portion. Note the short title for running heads
\title{Approximately Strategyproof Tournament Rules: On Large Manipulating Sets and Cover-Consistence}
\author[1]{Ariel Schvartzman  \thanks{acohenca@cs.princeton.edu.}}
\author[1]{S. Matthew Weinberg \thanks{smweinberg@princeton.edu. This work is supported by the National Science Foundation, under grant CNS-0435060, grant CCR-0325197 and grant EN-CS-0329609.}} 
\author[1]{Eitan Zlatin \thanks{ezlatin@princeton.edu.}}
\author[2]{Albert Zuo \thanks{azuo@stanford.edu.}} 
\affil[1]{Department of Computer Science, Princeton University.}
\affil[2]{Computer Science Department, Stanford University.}

% note that the abstract must come before \maketitle
% note: this command has been disabled to remove the ACM copyright block. Sorry...

\maketitle

\maketitle

\begin{abstract}
We consider the manipulability of tournament rules, in which $n$ teams play a round robin tournament and a winner is (possibly randomly) selected based on the outcome of all $\binom{n}{2}$ matches. Prior work defines a tournament rule to be $k$-$\SNM$-$\alpha$ if no set of $\leq k$ teams can fix the $\leq \binom{k}{2}$ matches among them to increase their probability of winning by $>\alpha$ and asks: for each $k$, what is the minimum $\alpha(k)$ such that a Condorcet-consistent (i.e. always selects a Condorcet winner when one exists) $k$-$\SNM$-$\alpha(k)$ tournament rule exists?

A simple example witnesses that $\alpha(k) \geq \frac{k-1}{2k-1}$ for all $k$, and~\cite{SchneiderSW17} conjectures that this is tight (and prove it is tight for $k=2$). Our first result refutes this conjecture: there exists a sufficiently large $k$ such that no Condorcet-consistent tournament rule is $k$-$\SNM$-$1/2$. Our second result leverages similar machinery to design a new tournament rule which is $k$-$\SNM$-$2/3$ for all $k$ (and this is the first tournament rule which is $k$-$\SNM$-$(<1)$ for all $k$). 

Our final result extends prior work, which proves that single-elimination bracket with random seeding is $2$-$\SNM$-$1/3$~\cite{SchneiderSW17}, in a different direction by seeking a stronger notion of fairness than Condorcet-consistence. We design a new tournament rule, which we call Randomized-King-of-the-Hill, which is $2$-$\SNM$-$1/3$ and \emph{cover-consistent} (the winner is an uncovered team with probability $1$). 
\end{abstract}

\newpage

% !TeX root = main.tex
% !TEX root = main.tex

\section{Introduction}
\label{sec:intro}
Consider $n$ teams vying for a single championship via pairwise matches. A \emph{Tournament Rule} maps (possibly randomly) the outcome of all $\binom{n}{2}$ matches to a single winner. A successful tournament rule should on one hand be fair, in that it selects a team that could reasonably be considered the best. For example, a tournament rule is Condorcet-consistent if whenever an undefeated team exists, that team is selected as the winner with probability $1$. On the other hand, a successful tournament rule should also incentivize teams to win their matches. For example in a monotone tournament rule, no team can unilaterally increase their probability of winning by throwing a match.

While numerous rules satisfy the two specific properties mentioned above, these hardly suffice to call a tournament both fair and incentive compatible. Consider for example a single-elimination tournament, which is Condorcet-consistent and monotone. One might reasonably argue that single-elimination is unfair in the sense that a \emph{covered} team may win. That is, some team $x$, who is beaten by $y$, and for which all $z$ who beat $y$ also beat $x$, could be crowned the champion (Observation~\ref{obs:RSEBcover}), even though $y$ is in some sense clearly a superior team. When multiple teams come from the same organization (e.g. in the Olympics where multiple teams from the same country participate, or in eSports where the same organization sponsors multiple teams), one could also argue that single-elimination is not incentive compatible: two teams from the same organization may wish to fix the match between them so that the team with the best chance of winning gold advances. 

Prior work establishes, however, that this stronger notion of incentive compatibility (termed $2$-Strongly Nonmanipulable by~\cite{AltmanK10}, and previously Pairwise Nonmanipulable in~\cite{AltmanPT09}) is incompatible even with the basic notion of Condorcet-consistency: no $2$-$\SNM$ tournament rule is Condorcet-consistent (recapped in Lemma~\ref{lem:sswlb}). This motivated~\cite{AltmanK10} to seek instead tournaments that were $2$-$\SNM$ and \emph{approximately} Condorcet-consistent (i.e. guaranteed to pick an undefeated team with probability at least $\alpha>0$, whenever one exists), and later~\cite{SchneiderSW17} to seek tournaments which were Condorcet-consistent and \emph{approximately} $2$-$\SNM$ (i.e. the maximum probability with which two teams can improve their joint probability of winning by fixing a match is $\alpha < 1$, termed $2$-$\SNM$-$\alpha$). 

Like~\cite{SchneiderSW17}, we find it more reasonable to seek a tournament which is only approximately strategyproof rather than one which is only approximately Condorcet-consistent: it is hard to imagine a successful sporting event which sends an undefeated team home empty-handed. The main result of~\cite{SchneiderSW17} proves that a Single-Elimination Bracket with Random seeding (\RSEB, formally defined in Section~\ref{sec:notation}) is both Condorcet-consistent and $2$-$\SNM$-$1/3$. This is tight, as no Condorcet-consistent tournament is $2$-$\SNM$-$\alpha$ for any $\alpha < 1/3$. They also define a tournament to be $k$-$\SNM$-$\alpha$ if no set of $\leq k$ teams can fix the $\leq \binom{k}{2}$ matches between them and improve their joint probability of winning by $\alpha$, establish that no Condorcet-consistent tournament rule is $k$-$\SNM$-$\alpha$ for $\alpha < \frac{k-1}{2k-1}$ (recapped in Lemma~\ref{lem:sswlb}), and conjecture that this is tight. The main open problem posed in their work is to prove this conjecture (recapped in Question~\ref{q:1}). The main results of this paper extend~\cite{SchneiderSW17} in three different directions:
\begin{itemize}
\item First, we resolve the main open problem posed in~\cite{SchneiderSW17} by \emph{refuting} their conjecture (including two weaker forms): There exists a sufficiently large $k$ such that no Condorcet-consistent tournament rule is $k$-$\SNM$-$1/2$ (and therefore not $k$-$\SNM$-$\frac{k-1}{2k-1}$ either). 
\item Second, we develop a new Condorcet-consistent tournament rule which is $k$-$\SNM$-$2/3$ for all $k$. All tournament rules are trivially $k$-$\SNM$-$1$ for all $k$, and this is the first tournament rule known to be $k$-$\SNM$-$\alpha$ for all $k$, for any $\alpha < 1$.
\item Finally, we develop a new \emph{cover-consistent} tournament rule which is $2$-$\SNM$-$1/3$, which we call Randomized King-of-the-Hill (\RKOTH). 
\end{itemize}

\subsection{Theorem Statements, Roadmap, and Technical Highlights}
\label{subsec:main}
After overviewing related work in Section~\ref{sec:rel}, and establishing preliminaries in Section~\ref{sec:notation}, we develop machinery related to our first two results in Section~\ref{sec:LP}. Specifically, consider the following: there are two kinds of manipulations which may cause a set $S$ of teams to improve their joint probability of winning under a Condorcet-consistent tournament rule $r(\cdot)$. First, perhaps $S$ contains $i$ and every team which beats $i$ (for some $i$). Then $S$ can make $i$ a Condorcet winner, and improve their joint probability of winning to $1$ (the maximum possible). Or, perhaps every team in $S$ loses to at least one team outside $S$, but $S$ can increase their joint probability of winning anyway.

If, for a particular tournament graph $T$, we wish to ensure that no set $S$ can improve their joint probability of winning by more than $\alpha$ under $r(\cdot)$ by creating a Condorcet winner in $S$, this is a fairly simple linear constraint: we must only ensure that for each team $i$, the joint probability of winning (under $r(T)$) of $i$ and all teams that beat $i$ in $T$ is at least $1-\alpha$. For each tournament graph $T$, this reasoning gives a feasibility linear program (with $n$ variables corresponding to the probability that each team wins under $r(T)$, and $n$ linear constraints which depend on $T$) that every winning-probability vector $r(T)$ must satisfy in order for $r(\cdot)$ to possibly be $k$-$\SNM$-$\alpha$. Note that these constraints are by no means \emph{sufficient} to guarantee $k$-$\SNM$-$\alpha$, however, as we have completely ignored the second type of constraints. 

Inspired by this feasibility LP, we study a similar LP in Section~\ref{sec:LP}, which we call the \emph{Special Linear Program} (\SLP). In particular, we show that \SLP\ has a unique solution, and therefore well-defines a tournament rule (if the outcome of the matches is $T$, solve the \SLP\ parameterized by $T$ and select according to these probabilities). Surprisingly, this LP and its unique optimal solution have been studied decades ago~\cite{LLL1993, FisherR92} (see Section~\ref{sec:rel} for further discussion of this --- the LP describes the unique Nash equilibrium of a generalized rock-paper-scissor game). In Section~\ref{sec:LB}, we explain why the \SLP\ Tournament Rule is special: \textbf{if \emph{any} tournament rule is $k$-$\SNM$-$1/2$ for all $k$, then the \SLP\ Tournament Rule is $k$-$\SNM$-$1/2$ for all $k$.} The remainder of Section~\ref{sec:LB} is then just a simple six team example witnessing that the \SLP\ Tournament Rule is \emph{not} $3$-$\SNM$-$1/2$ (and therefore not $k$-$\SNM$-$1/2$ for all $k$), yielding our first main result:

\begin{theorem}
\label{thm:newLB}
There exists $k < \infty$ such that no Condorcet-consistent tournament rule is $k$-$\SNM$-$1/2$. 
\end{theorem}

Note that Theorem~\ref{thm:newLB} implies that (a) no Condorcet-consistent tournament rule is simultaneously $k$-$\SNM$-$\frac{k-1}{2k-1}$ for all $k$, (b) there exists a $k$ for which no Condorcet-consistent tournament rule is $k$-$\SNM$-$\frac{k-1}{2k-1}$, and (c) no Condorcet-consistent tournament rule is $k$-$\SNM$-$1/2$ for all $k$, thereby refuting the main conjecture of~\cite{SchneiderSW17} along with two weaker conjectures. 

We also wish to emphasize the following: in principle, if one wishes to determine whether a particular tournament rule is $k$-$\SNM$-$\alpha$ for tournaments of $n$ teams, one could do an exhaustive search over all $2^{\binom{n}{2}}$ tournament graphs (with some savings due to isomorphism), and all $\binom{n}{k}$ possible manipulating sets. This is a feasible search for small values of $n, k$. If one wishes to determine whether there \emph{exists} a rule that is $k$-$\SNM$-$\alpha$ for tournaments of $n$ teams, one could still imagine an exhaustive search. But observe that the space of tournament \emph{rules} for $n$ teams lies in $n2^{\binom{n}{2}}$-dimensional space, and it is hard to imagine a successful exhaustive search beyond $n=10$ (and even that is likely impossible). 

Our proof establishes that no $939$-$\SNM$-$1/2$ tournament rule exists for $n=1878$ teams, and there is no hope of discovering this via exhaustive search. Indeed, our own exhaustive searches found numerous candidate rules which were $k$-$\SNM$-$1/2$ on $n$ teams for $k,n \leq 7$ (indicating that large parameters are provably necessary before the conjecture is false). Yet, the \SLP\ Tournament Rule is already not $3$-$\SNM$-$1/2$ even for $6$ teams (which we found via exhaustive search), and the machinery of Sections~\ref{sec:LP} and~\ref{sec:LB} allows us to use this tiny example to conclude that no $939$-$\SNM$-$1/2$ tournament rules exist for $n=1878$ teams. 

In Section~\ref{sec:UB}, we show how to use the machinery developed in Section~\ref{sec:LP} to propose new tournament rules. The main idea is the following: the \SLP\ Tournament Rule, by design, does a great job discouraging manipulations that produce a Condorcet winner. In fact, no such manipulation benefits the manipulators by more than $1/2$. Unfortunately, the rule itself may still be only $k$-$\SNM$-$1$ (due to manipulations which don't produce a Condorcet winner). However, we show that a convex combination of the \SLP\ Tournament Rule with the simple rule which selects a Condorcet winner when one exists, or a uniformly random winner otherwise can leverage the \SLP\ properties to be less manipulable:

\begin{theorem}
\label{thm:infUB}
There exists a Condorcet-consistent tournament rule that is $k$-$\SNM$-$2/3$ for all $k$. 
\end{theorem}

\begin{observation}
The rule referred to in Theorem~\ref{thm:infUB} is \emph{not} monotone. We provide an example of its non-monotonicity in Appendix~\ref{ubapp} for completeness.
\end{observation}

Finally, in Section~\ref{sec:cover}, we shift gears and extend the results of~\cite{SchneiderSW17} in a different direction. Specifically, we design a new tournament rule called Randomized King-of-the-Hill (\RKOTH) which is cover-consistent and $2$-$\SNM$-$1/3$. Each round, the rule checks if there is a team that is a Condorcet winner. If there is one, it declares that team the winner and terminates. Otherwise, it selects a uniformly random ``prince" among the remaining teams, and then removes it and every team which the prince beats. When there is only one team left, that team is crowned champion. It is not hard to see that \RKOTH\ is cover-consistent, so the main result of this section is that \RKOTH\ is $2$-$\SNM$-$1/3$. The main idea in the proof is that the joint winning probability of $\{u,v\}$ when $v$ beats $u$ is \emph{only} higher than when $u$ beats $v$ if $u$ is selected as a prince while $v$ still remains in contention, and at least one team which beats either $u$ or $v$ remains in contention as well. We are then able to show that the probability of this event is at most $1/3$, and therefore the rule is $2$-$\SNM$-$1/3$. 

\begin{theorem}
\label{thm:koth}
Randomized King-of-the-Hill is cover-consistent and $2$-$\SNM$-$1/3$. 
\end{theorem}

In fact, $\RKOTH$ satisfies a property stronger than cover-consistency. It will provably pick teams that belong to the \emph{Banks set} of the tournament, which may be strictly smaller than the set of uncovered teams. We defer the definition of the Banks and related proofs set to Section~\ref{sec:cover} but point out that the support of $\RKOTH$ is \emph{exactly} the Banks set of a tournament.

\begin{lemma}
\label{lem:bankskoth}
A team $v$ is in the Banks set of a tournament $T$ if, and only if, $\RKOTH$ declares $v$ as the winner with non-zero probability. 
\end{lemma}

\subsection{Extensions and Brief Discussion}
While the main appeal of our results is clearly theoretical (it is hard to imagine a $939$-team coalition manipulating a real tournament), the events motivating a deep study of fair and incentive compatible tournament rules are not purely hypothetical. In the popular ``group stage'' format, strategic manipulations have occurred on the grandest stage, including Badminton at the 2012 Olympics and the ``disgrace of Gij\'{o}n'' in the 1982 World Cup. But narrowing one's focus exclusively to incentives (and, e.g., running a single-elimination bracket) may have negative consequences for the quality of winner selected. For example in the 2010 World Cup, eventual winners and second seed Spain lost their opening match to Switzerland, who didn't advance out of the group stage (the implication being that Spain could be considered a ``high quality winner'' who would have been immediately eliminated in a single-elimination bracket). So while our particular results are valuable mostly for their theoretical contributions, the surrounding literature provides valuable insight on the tradeoff between incentive compatibility of the winner selection process and quality of the winner selected.

We also wish to briefly note that while we formally define SNM only for deterministic tournaments (i.e. the teams try to manipulate from a fixed tournament graph $T$ to another fixed tournament graph $T'$),~\cite{SchneiderSW17} establishes that all results extend to (arbitrarily correlated) distributions over tournament graphs as well (where the ``real'' outcomes may be a distribution $\mathcal{T}$ over tournament graphs, and the coalition may try to manipulate to a different distribution $\mathcal{T}'$). We refer the reader to~\cite{SchneiderSW17} for a formal statement, but the main idea is that any lower bounds immediately carry over, while for every rule the largest possible manipulation occurs on a deterministic instance anyway (so positive results carry over as well). 

\subsection{Related work}\label{sec:rel}
The most related works have already been discussed above:~\cite{AltmanPT09} first introduces the terminology used for these problems (and establishes that no deterministic tournament rule is $2$-$\SNM$),~\cite{AltmanK10} first considers randomized tournament rules and designs tournament rules which are $2$-$\SNM$ but only approximately Condorcet-consistent.~\cite{SchneiderSW17} is the most closely related, which also considers rules which are Condorcet-consistent and approximately incentive compatible. Our work can most appropriately be viewed as extending~\cite{SchneiderSW17} in multiple directions (including resolving their main open problem) as detailed in Section~\ref{sec:intro}. 

Also related are some recent works which rigorously analyze the manipulability of specific tournament formats (most notably, the World Cup and related qualifying procedures)~\cite{Pauly14, Csato17}. 

Incentive compatibility of \emph{voting rules} has an enormous history, dating back at least to seminal works of~\cite{Arrow50, Gibbard73, Satterthwaite75, Gibbard77}. While there are obvious conceptual connections between voting rules and tournament rules (e.g. any tournament rule can be used as a voting rule: call the ``match'' between alternatives $x$ and $y$ won by $x$ if more voters prefer $x$ to $y$), the notions of manipulability are quite different. For one, voters have preferences over alternatives whereas teams in tournaments only care about their collusion's joint probability of winning. Moreover, in a voting rule, a \emph{voter} has a tiny role to play in every single ``match,'' whereas in a tournament, the \emph{teams} themselves can manipulate only matches that involve them. So there is little technical (and even conceptual) similarity between works which study incentives in voting rules versus tournament rules. 

The linear program we introduce in Section~\ref{sec:LP} has been studied in a related context by two independent works~\cite{LLL1993, FisherR92}. They consider the following two player zero-sum game. Given a tournament $T$ the players must pick a team to represent them and reveal it simultaneously. If they pick the same team, no one wins. Otherwise, the winner is determined by the edge that they jointly query from $T$. The two works showed that, for any given tournament $T$, there is a unique mixed strategy Nash equilibrium which can be computed in polynomial time through linear programming (and this LP is exactly our $\SLP(T)$). Some proof techniques are similar, but our interest in $\SLP(T)$ is a means to drastically different end. For more on the history and properties of the solution to these LPs, known as \emph{maximal lotteries}, see~\cite{Fishburn84, BrandlBS15}.

The notion of uncovered teams is also extremely well-studied in computational social choice theory (see, e.g.,~\cite{BCELP16, Laslier1997}, the latter attributes the concept's introduction to~\cite{Fishburn77} and~\cite{Miller80} independently). Additionally, an uncovered team is equivalent to the notion of a ``king''~\cite{ShepsleW84} (a team $x$ such that for all teams $y$, either $x$ beats $y$ or there exists a $z$ who beats $y$ such that $x$ beats $z$ -- not to be confused with the kings of our hill) in works which study how a single-elimination bracket designer can rig the seeding to make a particular team win~\cite{StantonW11, KimW15, Williams10} or sufficient conditions under which a covered team can be crowned winner of a single-elimination bracket~\cite{KimSW16}. The volume of these works certainly helps establish that cover-consistence is a valuable endeavor beyond Condorcet-consistence, but otherwise bear no technical similarity to our work. To the best of our knowledge, this is the first paper to consider cover-consistence jointly with a notion of incentive compatibility.

Another related notion is that of the Banks set of a tournament~\cite{Banks1985}, which is stricter than the set of uncovered teams. While it is NP-Complete to decide if a given team is in the Banks set of a tournament~\cite{Woeginger2003}, there exist algorithms that can efficiently output an element from the Banks set~\cite{Hudry2004}. It is worth pointing out that the algorithm we propose to sample teams from the Banks set, Algorithm~\ref{alg:2snm}, has a different implementation from that of~\cite{Hudry2004} even if they will output the same set of teams.

\section{Preliminaries}
\label{sec:notation}

In this section we introduce notation, and develop some concepts that will be relevant throughout the paper, consistent with prior work~\cite{SchneiderSW17, AltmanK10}. 

\begin{definition}[Tournament]
A (round robin) tournament $T$ on $n$ teams is a complete, directed graph on $n$ vertices whose edges denote the outcome of a match between two teams. Team $i$ beats team $j$ if the edge between them points from $i$ to $j$. $T_n$ denotes the set of all $n$-team tournaments.
\end{definition}

\begin{definition}[Tournament Rule]
A tournament rule $r$ is a function $r: T_n \rightarrow \Delta([n])$ that maps $n$-team tournaments $T \in T_n$ to a distribution over teams, where $r_i(T) = \Pr(r(T) = i)$ denotes the probability with which team $i$ is declared the winner of tournament $T$ under rule $r$. We will often abuse notation and refer to $r$ as a collection of tournament rules $\{r^1(\cdot),\ldots, r^n(\cdot),\ldots\}$, of which exactly one operates on $T_n$ (for all $n$). 
\end{definition}

Like prior work, we will be interested in tournaments which satisfy natural properties. For instance,~\cite{SchneiderSW17} concerned tournaments which always select a Condorcet winner, when one exists. Below are the main properties we consider in this paper.

\begin{definition}[Condorcet-consistency]
Team $i$ is a Condorcet winner of a tournament $T$ if $i$ beats every other team according to $T$. A tournament rule $r$ is Condorcet-consistent if for every tournament $T$ with a Condorcet winner $i$, $r_i(T) = 1$ (i.e. the tournament rule always declares the Condorcet winner as the winner of $T$). 
\end{definition}

\begin{definition}[cover-consistency]
Team $i$ covers team $j$ under $T$ if (a) $i$ beats $j$ and (b) every $k \notin\{i,j\}$ which beats $i$ also beats $j$. A team is covered if it is covered by at least one team. A tournament rule $r$ is cover-consistent if for all $T$, $r_j(T) = 0$ when $j$ is covered.  
\end{definition}

\begin{observation} Every cover-consistent rule is Condorcet-consistent.
\end{observation}
\begin{proof}
If $T$ has no Condorcet winner, then a Condorcet-consistent rule can be aribtrary on $T$. If $T$ has a Condorcet winner $i$, then $i$ is the only uncovered team. Therefore, any cover-consistent rule will have $r_i(T) = 1$, and is Condorcet-consistent as well.
\end{proof}

Intuitively, one should think of $i$ covering $j$ to mean that any reasonable evaluation should declare team $i$ better than team $j$. Cover-consistence proposes that no team should win if they are inferior to another by any reasonable evaluation, but does not always propose who the winner should be. Condorcet-consistence can therefore be interpreted as a relaxation of cover-consistence, which only binds when cover-consistence would propose a unique winner. The following lemma establishes this formally.

\begin{lemma}[Restated Theorem 4 from~\cite{Maurer80}] Whenever tournament $T$ has a unique uncovered team $i$, $i$ is a Condorcet winner in $T$.
\end{lemma}
\begin{proof}
We first claim that the covering relation is transitive: if $i$ covers $j$ and $j$ covers $k$, then $i$ covers $k$. Indeed, let $S(k)$ denote the teams that $k$ defeats, $S(j)$ for $j$, and $S(i)$ for $i$. Then as $j$ covers $k$ and $i$ covers $j$, we have $S(k) \cup \{k\} \subseteq S(j) \subseteq S(i)$, meaning that $i$ covers $k$. 

Therefore, if we draw the directed graph $G(T)$ with an edge from $j$ to $k$ iff $j$ covers $k$, the graph must be acyclic. If not, then the work above establishes that a path from $j$ to $k$ implies that $j$ covers $k$, while a path from $k$ to $j$ implies that $k$ covers $j$, a contradiction (as we cannot have both that $j$ beats $k$ and $k$ beats $j$). Uncovered teams are exactly those with indegree $0$ in $G(T)$. If there is a unique such team $i$, then there must be a path from $i$ to every other team $j$ (follow edges backwards starting from $j$. This process must terminate, and can only terminate at a team with indegree $0$, which must be $i$). Therefore, $i$ covers every other team $j$, which in particular implies that $i$ beats every other team $j$ and is therefore a Condorcet winner.
\end{proof}

The above conditions concern natural properties of the winner selected, and essentially say that good tournament rules should never select obviously inferior teams as their winner. We are also concerned with properties regarding the procedure by which the winner is selected, and in particular how manipulable this procedure is. We formalize these properties below (which are proposed in~\cite{AltmanK10, SchneiderSW17}). 

\begin{definition}[$S$-Adjacent Tournaments]
Two tournaments $T, T' \in T_n$ are $S$-adjacent when the outcomes of all matches in $T, T'$ are identical, except for the matches between two teams in $S \subseteq [n]$. Formally, for all $i, j \in [n]$, if $|\{i,j\} \cap S|<2$, then the edge between $i$ and $j$ in $T$ is identical to the one in $T'$. Less formally, $T$ and $T'$ are $S$-adjacent if teams in $S$ can manipulate the outcomes of matches \emph{only} between pairs of teams in $S$ and cause the tournament results to change from $T$ to $T'$.
\end{definition}

\begin{definition}[$k$-$\SNM$-$\alpha$]
A tournament rule $r$ is $k$-strongly non-manipulable at probability $\alpha$ (henceforth $k$-$\SNM$-$\alpha$) if for all subsets $S \subseteq [n]$ with $|S| \leq k$, and all pairs $T, T'$ of $S$-adjacent tournaments, $\sum_{i \in S} (r_i(T)-r_i(T')) \leq \alpha$. Informally, if a set $S$ of $\leq k$ teams decide to manipulate their pairwise matches, they cannot improve the probability that the winner is in $S$ by more than $\alpha$. We abuse notation and use $\infty$-$\SNM$-$\alpha$ to refer to a tournament rule which is $k$-$\SNM$-$\alpha$ for all $k$.
\end{definition}

\subsection{Technical Recap of Prior Work}

Finally, let's recap~\cite{SchneiderSW17}, which serves as the starting point for our work. Their main result establishes that the Random Single-Elimination Bracket (formally defined below) is $2$-$\SNM$-$1/3$. 

\begin{definition}[Random Single-Elimination Bracket]
The Random Single-Elimination Bracket ($\RSEB$) rule places $n$ teams uniformly at random among $2^{\lceil \log_2 n \rceil}$ seeds \footnote{A \emph{seed} is a position in the tournament bracket associated to a specific number.} (and fills the remaining seeds with byes)\footnote{A \emph{bye} is a dummy team that loses to all non-bye teams.	}. Then, a single-elimination tournament is played with these seeds to determine the winner. That is, whenever $i$ meets $j$ in the bracket, $T$ determines which team advances to the next round, and the other team is eliminated. Note that the only randomness in the rule is in the seeding.
\end{definition}

\begin{theorem}[\cite{SchneiderSW17}]\label{thm:SchneiderSW} $\RSEB$ is $2$-$\SNM$-$1/3$ and Condorcet-consistent.
\end{theorem}

The following explicit tournament was also used in~\cite{SchneiderSW17} for lower bounds:

\begin{definition}[Balanced Tournament]
\label{def:balT}
The $k$-balanced tournament is the tournament $\bal \in T_{2k-1}$ where team $i$ beats exactly the $k-1$ teams in $\{ i+1, i+2, ..., i+k-1 \mod (2k-1) \}$. 
\end{definition}

\begin{lemma}[\cite{SchneiderSW17}]\label{lem:sswlb}
No Condorcet-consistent Tournament rule is $k$-$\SNM$-$\alpha$ for any $\alpha < \frac{k-1}{2k-1}$. 
\end{lemma}
\begin{proof}[Proof Sketch]
Consider $r(\bal)$. There exists some adjacent set of teams $S = \{i-k+1\pmod{2k-1},\ldots, i\}$ of size $k$ which together win with probability at most $\frac{k}{2k-1}$ in $r(\bal)$. These teams can make $i$ into a Condorcet winner, which necessarily wins with probability $1$. Therefore, for any $r(\cdot)$, some set of size $k$ can gain at least $\frac{k-1}{2k-1}$ by manipulating when the original tournament is $\bal$.
\end{proof}

Inspired by the tightness of Theorem~\ref{thm:SchneiderSW} with the simple balanced tournament $\bal$,~\cite{SchneiderSW17} conjectured that same simple tournament would be the worst-case tournament for larger $k$:

\begin{open}[\cite{SchneiderSW17}]\label{q:1} Does there exist a tournament rule that is Condorcet-consistent and $k$-SNM-$\frac{k}{2k-1}$ for all $k$? What about a family of rules $\mathcal{F}$ such that for all $k$, $F_k$ is $k$-SNM-$\frac{k-1}{2k-1}$? What about a rule that is $k$-SNM-$1/2$ for all $k$?
\end{open}

The first results of this paper refute all three conjectures from~\cite{SchneiderSW17} and resolve Question~\ref{q:1}. The following results concern the difference between Condorcet-consistence and cover-consistence, as the following observation shows that $\RSEB$ is not cover-consistent.

\begin{observation}\label{obs:RSEBcover} $\RSEB$ is not cover-consistent.
\end{observation}

\begin{proof}
Consider a tournament with eight teams $A, B, C, D, E, F, G, H$, where $A$ beats exactly $\{B, C, E\}$, and $H$ beats exactly $\{A, B, C, E\}$. $C$ beats $D$, $E$ beats $F$, $E$ beats $G$. Any matches not explicitly stated can be arbitrarily decided. Consider the seeded bracket shown in Figure~\ref{fig:coveredWin}. This bracket shows $A$ can win with non-zero probability. But $H$ covers $A$. Therefore, $\RSEB$ is not cover-consistent.
\end{proof}

\begin{figure}
\centering
\includegraphics[width=0.75\textwidth]{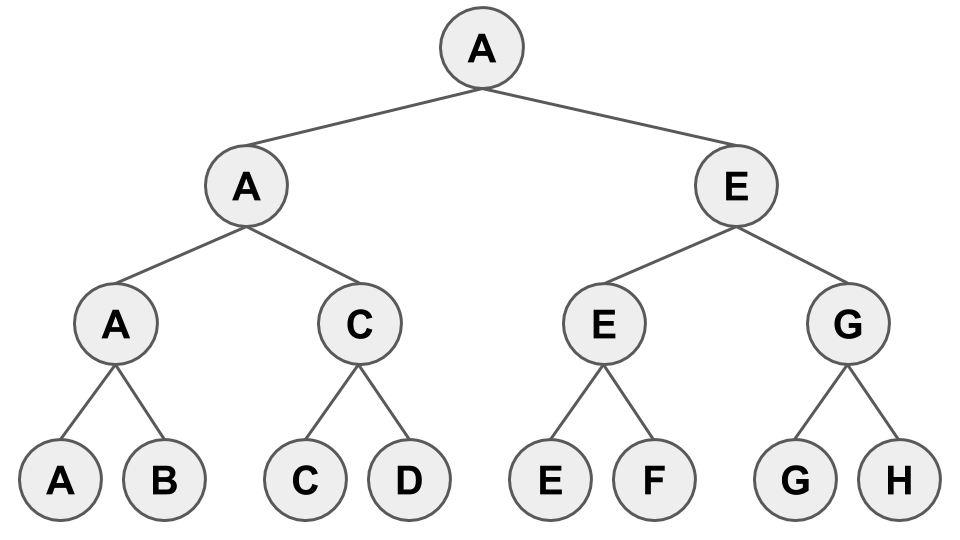}
\caption{In this example there exists a bracket where a covered team $A$ (covered by $H$) may still be declared the winner.}
\label{fig:coveredWin}
\end{figure}

\subsection{Linear Algebra Preliminaries}\label{sec:linalg}
Some of our proofs require linear algebra. Below are facts that we use, both proofs are in Appendix~\ref{app:prelim}.

\begin{definition}[Unit Skew Symmetric Matrix]
An $n \times n$ matrix $A \in \mathbb{R}^{n \times n}$ is \emph{unit skew symmetric} if $|A_{ij}| = 1$ $\forall i\ne j$, and $A_{ij} = - A_{ji}$ $\forall i,j$. 
\end{definition}

\begin{proposition}
	\label{prop:skew-symmetric-rank}
	If $A$ is a unit skew symmetric matrix and $n$ is even, $\text{rank}(A) = n$. If $n$ is odd $\text{rank}(A) = n - 1$.
\end{proposition}

\begin{proposition}
	\label{prop:epsilon-delta}
	Let $\varepsilon \in \mathbb{R}_{\geq 0}$, $A \in \mathbb{R}^{m \times n}$, and $\vec{b} \in \mathbb{R}^m$ and denote by $P_{\varepsilon}:= \{\vec{x} \in \mathbb{R}^n,  A\cdot \vec{x} \geq \vec{b}- \varepsilon \vec{1}\} \cap [0,1]^n$. Then for all $\delta > 0$, there exists a sufficiently small $\varepsilon>0$ such that:
$$\max_{\vec{y} \in P_{\varepsilon}}\{d_{\ell_1}(P_0,\vec{y})\} \leq \delta , $$

	where $d_{\ell_1}(S,\vec{x}) = \min_{\vec{y} \in S} \{|\vec{x} - \vec{y}|_1\}.$
\end{proposition}

% !TeX root = main.tex
% !TEX root = main.tex
\section{A Special Linear Program $SLP(T)$}
\label{sec:LP}
In this section we present a linear program $SLP(T)$ and characterize its optimal solutions. The analysis of $SLP(T)$ is the main tool which allows us to conclude both the non-existence of rules which are $\infty$-$\SNM$-$1/2$ (Section~\ref{sec:LB}) and the existence of a rule which is $\infty$-$\SNM$-$2/3$ (Section~\ref{sec:UB}). The main result of this section is Proposition~\ref{prop:uniqueness}, which states that $SLP(T)$ has a unique solution, and therefore yields a well-defined tournament rule. In Section~\ref{sec:LB} we show that a tournament rule that  is $\infty$-$\SNM$-$1/2$ exists \emph{if and only if} the $SLP(T)$ rule is $k$-$\SNM$-$1/2$ (and subsequently show that this rule is not $k$-$\SNM$-$1/2$ via Proposition~\ref{prop:uniqueness}). Let $T$ be a tournament graph and let $\delta^{-}_T(v)$ denote the set of teams that beat $v$ in $T$ (and $\delta^{+}_T(v)$ the set that $v$ beats). Then $SLP(T)$ is the following:\\
\\
\noindent\textbf{$SLP(T)$:}
\vspace{-5mm}
\begin{align*}
\text{minimize} \sum_{i=1}^{n} p_i \\
\text{subject to} \sum_{j \in \delta^{-}_T(i)} &p_{j} + \frac{1}{2}p_i \geq \frac{1}{2}  \hspace{2cm} &&\forall i \in [n]\\
& p_i \geq 0 \hspace{2cm} &&\forall i \in [n]
\end{align*}

Before further proceeding, let's get some (informal) intuition for why $SLP(T)$ is possibly related to Question~\ref{q:1}. Starting from a tournament rule $r$, if we define $p_i:= r_i(T)$, then $\sum_i p_i = 1$, $p_i \geq 0$ for all $i$. Moreover, if $r$ is $\infty$-$\SNM$-$1/2$, it must be that for all $i$, $\sum_{j \in \delta^-_T(i)} p_j + p_i \geq 1/2$. If not, then $i$ together with $\delta^-_T(i)$ could collude to make $i$ a Condorcet winner, and $i$ would win with probability $1$. So the initial probability of winning for $i$ together with $\delta^-_T(i)$ must have been at least $1/2$. 

Of course, the afore-described constraints seem very weak in comparison to all of the constraints imposed by $k$-$\SNM$-$1/2$. In particular, they only guarantee that no coalition can gain by making one of their members into a Condorcet winner (but do not guarantee that no coalition can otherwise gain by manipulating their matches). Notice now that the constraints in $SLP(T)$ are slightly stronger than this (because they have a multiplier of $1/2$ instead of $1$ in front of $p_i$ in the constraint for $i$). In particular, the constraints in $SLP(T)$ imply Condorcet-consistence (while the afore-mentioned do not): if $i$ is a Condorcet winner, then $\delta^-_T(i) = \emptyset$ and the constraint reads $p_i/2\geq 1/2$ as desired. Of course, we've yet to establish a formal relationship, but at this point the reader may have some intuition for a connection between a profile of solutions to $SLP(T)$ with $\sum_i p_i \leq 1$ (for all $T$) and Condorcet-consistent tournament rules which are $\infty$-$\SNM$-$1/2$. 

We postpone a formal discussion of this connection (as this connection is the entire focus of Section~\ref{sec:LB}), but note here that it is not particularly direct. For example, a profile of solutions to $SLP(T)$ for all $T\in T_n$ does not imply a tournament rule for $n$ teams which is $\infty$-$\SNM$-$1/2$. Similarly, an $\infty$-$\SNM$-$1/2$ tournament rule for $n$ teams does not imply a profile of solutions to $SLP(T)$ for all $T\in T_n$. However, we show that $\infty$-$\SNM$-$1/2$ rules exist for all $n$ \emph{if and only if} for all $n$, the rule defined via profiles of solutions to $SLP(T)$ is $\infty$-$\SNM$-$1/2$ (i.e. we will relate this LP on $n$ teams to tournament rules for $\gg n$ teams). 
%The dual for $SLP(T)$, call it $D_{SLP}(T)$, is the following. 

%\begin{align*}
%	\text{maximize} \displaystyle\sum\limits_{j=1}^{n} \frac{w_j}{2} \\
%	\text{subject to} \displaystyle\sum\limits_{j \in \delta^{-}_T(v)} &w_{j} + \frac{1}{2}w_v \leq 1  \hspace{2cm} &&\forall v \in [n]\\
%	& w_v \geq 0 \hspace{2cm} &&\forall v \in [n]
%\end{align*}
 
%To make some of the proofs easier, I will perform the relabeling of variables $q_i = \frac{w_i}{2}$ and multiply each constraint by $\frac{1}{2}$ to get an LP equivalent to $D_{SLP}(T)$. 

We now begin our analysis of $SLP(T)$ by taking the dual, and refer to it as $D_{SLP}(T)$. Below, we use $r_i$ as the dual variable for the constraint corresponding to team $i$. On the left-hand side, we've taken the dual directly. On the right hand side, we did a change of variables and redefined $q_i := r_i/2$ (so the two programs below are identical).

\noindent\textbf{$D_{SLP}(T)$:} 
\vspace{-5mm}
\begin{alignat*}{8}
&\text{maximize} \hspace{2.5mm} \displaystyle \sum\limits_{i=1}^n &&r_i/2 
 &&\ &&\ &&\text{maximize}\hspace{2.5mm} \displaystyle\sum\limits_{i=1}^{n} &&q_i &&\ &&\\
&\text{subject to} \displaystyle\sum\limits_{j \in \delta^{+}_T(i)} &&r_{j} + \frac{1}{2}&&r_i&&\leq 1 \hspace{0.2cm} \forall i \in [n] \hspace{1cm}
&&\text{subject to}\displaystyle\sum\limits_{j \in \delta^{+}_T(i)} &&q_{j} + \frac{1}{2} &&q_i\leq \frac{1}{2} \hspace{0.2cm} \forall i \in [n]&&\\
&\qquad &&\qquad &&r_i&&\geq 0 \hspace{0.2cm} \forall i \in [n] \hspace{1cm}
&&\qquad &&\qquad  &&q_i\geq \ 0 \hspace{0.2cm}\forall i \in [n]&&
\end{alignat*}

We now prove that the optimal value of $SLP(T)$ is always $1$. This is stated in Corollary~\ref{cor:unitness}, which uses Lemma~\ref{lem:primal-dualconv} as a building block.

\begin{lemma}
\label{lem:primal-dualconv}
	Suppose there exists a feasible solution $\vec{p}$ to $SLP(T)$ with  $\sum_{i\in [n]} p_i = c$. Then $\vec{q}$ with $q_{i} := p_{i} \cdot \frac{1}{2c - 1}$ is a feasible solution to $D_{SLP}(T)$ with value $\frac{c}{2c - 1}$. Likewise, if there exists a feasible solution $\vec{q}$ to $D_{SLP}(T)$ with  $\sum_{i \in [n]} q_i = c$, then $\vec{p}$ with $p_{i} := q_{i} \cdot \frac{1}{2c - 1}$ is a feasible solution to $SLP(T)$ with value $ \frac{c}{2c - 1}$.
\end{lemma}

\begin{proof}
Consider any solution $\vec{p}$ with $\sum_i p_i = c$. First, we observe that we must have $c > 1/2$. If not, there certainly exists some $i$ with $\sum_{j \in \delta^-_T(i)} p_j + p_i/2 < 1/2$, and a constraint is violated (to see this, observe that maybe $c= 0$, in which case all the constraints are violated. Or $0 < c \leq 1/2$, in which case we can take $i$ to be any $i$ with $p_i>0$). Then because $\sum_{j \in \delta^-_T(i)} p_j +p_i/2 \geq 1/2$, we must have $\sum_{j \in \delta^+_T(i)} p_j + p_i/2 \leq c-1/2$. As $q_i:= p_i/(2c-1)$, we immediately conclude that $\sum_{j \in \delta^+_T(i)} q_j + q_i/2 \leq \frac{c-1/2}{2c-1} = 1/2$. Also, as $c > 1/2$, each $q_i \geq 0$ (and is well-defined). Therefore, $\vec{q}$ is feasible for $D_{SLP}(T)$, and it's clear that $\sum_i q_i = \frac{c}{2c-1}$. The other direction follows from identical calculations.
\end{proof}

%\begin{lemma}
%\label{dual-primalconv}
%	Suppose there exists a feasible solution $\vec{p}$ to $D_{SLP}(T)$ with  $\sum_{j \in [n]} q_j = c$. Then if we define $p_{j} = q_{j} \cdot \frac{1}{2c - 1}$, $\vec{p}$ is a solution to $SLP(T)$ with value $ \frac{c}{2c - 1}$.
%\end{lemma}
%The proof of the above lemma is very similar to that of the previous one, so I will omit it here.

\begin{corollary}
	\label{cor:unitness}
	$SLP(T)$ always has an optimal solution with value 1.
\end{corollary}

\begin{proof}
	It is clear that $SLP(T)$ is feasible for all $T$, since setting $p_i = 1$ for all $i$ is a feasible solution. Suppose we had a primal solution $\vec{p}$ with value $c < 1$. Applying Lemma~\ref{lem:primal-dualconv}, we can conclude $\vec{q}$ would be a dual solution with value $\frac{c}{2c -1} > c$. By weak LP duality, the existence of such a dual would verify that there are no primal solutions with value $c$, a contradiction. 

	Similarly, suppose we had an optimal primal solution $\vec{p}$ with value $c > 1$. This implies there is an optimal dual solution $\vec{q}$ with value $c > 1$. Applying the opposite direction of Lemma~\ref{lem:primal-dualconv} we can conclude there is a primal solution with value $\frac{c}{2c-1} < c$, a contradiction.
\end{proof}

Now that we know the optimal value of $SLP(T)$, we wish to understand its optimal solution. We now begin taking steps towards characterizing the solution (and in particular, that it is unique). 

\begin{corollary}
	\label{cor:complementary-slackness}
	Let $\vec{p}$ be an optimal solution to $SLP(T)$. Then for all $i$ s.t $p_i > 0$ and for all optimal solutions $\vec{w}$ to $SLP(T)$, $\sum_{j \in \delta^{-}_T(i)} w_{j} + \frac{1}{2}w_i = \frac{1}{2} $.
\end{corollary}

\begin{proof}
	If we apply Lemma~\ref{lem:primal-dualconv} to $\vec{p}$ (which has $|\vec{p}|_1 = 1$ by Corollary~\ref{cor:unitness}), we conclude $\vec{p}$ is also a feasible dual solution. Hence, $\vec{p}$ is an optimal dual solution (as it has equal value in both the primal and the dual). Now consider applying the complementary slackness conditions for the alternative optimal primal $\vec{w}$, and optimal dual solution $\vec{p}$. If $p_i > 0$, we know that the corresponding primal constraint in $\vec{w}$ must be tight. This exactly states that $\sum_{j \in \delta^{-}_T(i)} w_{j} + \frac{1}{2}w_i= \frac{1}{2}$ whenever $p_i > 0$. 
\end{proof}

\begin{proposition}
	\label{prop:uniqueness}
	The optimal solution $\vec{p}$ to $SLP(T)$ is unique.
\end{proposition}

\begin{proof}
	By Corollary~\ref{cor:unitness}, we know that all solutions to $SLP(T)$ have value $1$. Assume toward contradiction there exist two distinct solutions, $\vec{p}$ and $\vec{q}$ to $SLP(T)$ such that $|\vec{p}|_{1} = 1$ and $|\vec{q}|_{1} = 1$. Let $P = \{i: p_i > 0\}$ and let  $Q = \{i: q_i > 0\}$. Let $A$ be the unit skew symmetric matrix where $A_{ij} = 1$ if $i$ beats $j$ in T and $-1$ otherwise, with $A_{ii} = 0$. By Corollary~\ref{cor:complementary-slackness}, we know for all $i$ in $P \cup Q$ that 

$$\sum_{j \in \delta^{-}_T(i)} q_{j} + \frac{1}{2} q_i = 1/2$$

 which implies for those same $i \in Q \cup P$ that (because $\sum_j q_j = 1$) 

$$(A \cdot {q})_i = \sum_{j \in \delta^{+}_T(i)} q_{j} - \sum_{j \in \delta^{-}_T(i)} q_{j} = 0.$$

Now, let $A'$ denote the submatrix $A$ restricted only to rows and columns in $P \cup Q$. Let $\vec{p}'$ and $\vec{q}'$ be the vectors $\vec{p}$ and $\vec{q}$ (respectively) restricted also to the entries in $P \cup Q$. Observe now that $(A' \cdot \vec{p}')_i = (A \cdot \vec{p})_i = 0$ for all $i$, and also that $(A' \cdot \vec{q}')_i = (A \cdot \vec{q})_i = 0$ (both of these follow because we have simply deleted all non-zero entries of $A \cdot \vec{p}$ and $A \cdot \vec{q}$ by restricting to $P \cup Q$). 

Now we are ready to derive our contradiction. The above paragraph concludes that both $\vec{p}'$ and $\vec{q}'$ are in the null space of $A'$, which is a unit skew symmetric matrix. But also $|\vec{p}|_1 = |\vec{q}|_1$, meaning that the null space of $A'$ must have dimension at least $2$. But this contradicts Proposition~\ref{prop:skew-symmetric-rank}, which claims that the dimension can be at most $1$. We therefore conclude that no such distinct $\vec{p},\vec{q}$ can exist.
\end{proof}

Now that we know the solution to $SLP(T)$ is unique, and has $\sum_i p_i = 1$, it yields a well-defined tournament rule, which is the main takeaway from this section:

\begin{definition}[$\SLP$ Tournament Rule]
Let $\vec{p}(T)$ denote the (unique, by Proposition~\ref{prop:uniqueness}) solution to $SLP(T)$. Define the \emph{$\SLP$ Tournament Rule} to select $i$ as the winner with probability $p_i(T)$ on input $T$. 
\end{definition}

\section{No Condorcet-consistent $\infty$-SNM-$1/2$ Rule Exists}
\label{sec:LB}
In this section we leverage our analysis of $\SLP$ to prove Theorem~\ref{thm:newLB}. First, we make the connection between rules that are $k$-$\SNM$-$\alpha$ and $\SLP$, by introducing a series of linear programs and relaxations.

Recall that for a Condorcet-consistent tournament rule $r$ to be $k$-$\SNM$-$\alpha$, it must be that no coalition of size $k$ can gain more than $\alpha$ probability of winning by manipulating the pairwise matches between them. In particular, if $|\delta^-_T(v)|< k$, it must be that $r_v(T)+\sum_{j \in \delta^{-}_T(v)} r_j(T) \geq 1 - \alpha$. Otherwise the set $\delta^-_T(v) \cup \{v\}$ can collude to make $v$ a Condorcet winner. Formally, any $k$-$\SNM$-$\alpha$ rule must satisfy the following feasibility $LP_0(T, \alpha, k)$ for all tournaments $T$.\\
\\
\vspace{-5mm}
\begin{align*}
\label{LP0}
LP_0(T,\alpha,k): \\
 p_i+\sum_{j \in \delta^{-}_T(i)} p_{j} &\geq 1 - \alpha & \forall i \in [n] \text{ such that }|\delta^{-}_T(i)| \leq k-1\\
 \sum_{\forall v} p_i &= 1 & \\
 p_i &\geq 0  & \forall i \in [n]
\end{align*}

Note that any $k$-$\SNM$-$\alpha$ rule certainly satisfies $LP_0(T,\alpha,k)$ for all $T$, but that a profile of solutions to $LP_0$ for all $T \in T_n$ does not necessarily imply a rule which is $k$-$\SNM$-$\alpha$ (as the LP only considers deviations which produce a Condorcet winner). Note also that the $k$-balanced tournament witnesses that no rule satisfies $LP_0(T, \alpha, k)$ for any $\alpha < \frac{k-1}{2k-1}$. We will also consider the case where $k \rightarrow \infty$ (and therefore, all $i \in [n]$ have $|\delta^-_T(i)| \leq k-1$, and refer to this LP simply as $LP_0(T,\alpha):= LP_0(T, \alpha, \infty)$). 

Our first step will be switching from $LP_0(T, \alpha)$ to $LP_1(T, \alpha, z)$ for $z \geq 1$. Below, observe that we have made two changes. The first is insignificant: we've phrased $LP_1(T, \alpha, z)$ as a minimization LP instead of a feasibility LP. The second is a strengthening: we've changed the multiplier of $p_i$ in the constraint corresponding to $i$ from $1$ to $\frac{z}{2z-1} \leq 1$ (so the space of feasible solutions is smaller).\\
\\
\noindent\textbf{$LP_1(T,\alpha,z)$:}
\vspace{-5mm}
\begin{align*}
\text{minimize} \sum_{j=1}^{n} p_j \\
\text{subject to} \sum_{j \in \delta^{-}_T(i)} &p_{j} + \frac{z}{2z - 1}p_i \geq 1 - \alpha &&\forall i \in [n]\\
& p_i \geq 0 &&\forall i \in [n]
\end{align*}

Observe that $LP_1(T, 1/2 ,z)$ is a relaxation of $SLP(T)$ (the only difference is a multiplier of $\frac{z}{2z-1} > 1/2$ in front of $p_i$ in the constraint corresponding to $i$). The main step in this section is Lemma~\ref{lem:almost-half} below, which formally connects $LP_1(T, 1/2,z)$ to $k$-$\SNM$-$\alpha$ rules.

\begin{lemma} 
	\label{lem:almost-half}
	If for all $n$ there exists a tournament rule $r(\cdot)$ which is $\infty$-$\SNM$-$\alpha$, then for all $z \in \mathbb{N}_+$ and all $n$, there exists a tournament rule $w(\cdot)$ which is $\infty$-$\SNM$-$\alpha$ \emph{and for which $w(T)$ is a feasible solution to $LP_1(T,\alpha,z)$ with $\sum_i p_i=1$ for all $T$}. 

Similarly, if for all $n$ there exists a tournament rule $r(\cdot)$ which is $k$-$\SNM$-$\alpha$, then for all $z \in \mathbb{N}_+$ and all $n$, there exists a tournament rule $w(\cdot)$ which is $\frac{k}{2z-1}$-$\SNM$-$\alpha$ \footnote{We abuse notation throughout this section to define $\frac{k}{2z-1}$-$\SNM$-$\alpha$ as $\lfloor\frac{k}{2z-1}\rfloor$-$\SNM$-$\alpha$ when the first term may not be an integer.} and for which $w(T)$ is a feasible solution to $LP_1(T, \alpha, z)$ with $\sum_i p_i=1$ for all $T$.
\end{lemma}

\begin{figure}%
    \centering
    {{\includegraphics[width=7.5cm]{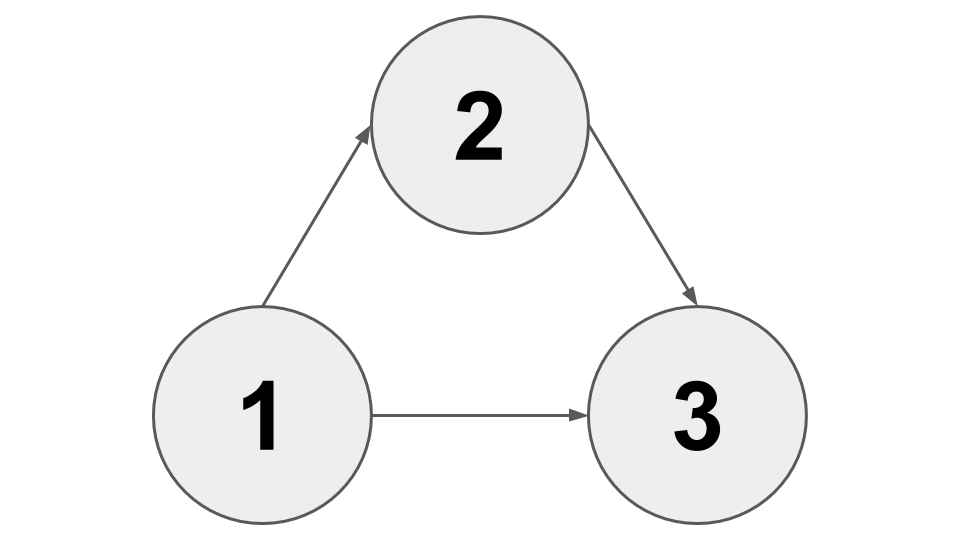} }}%
    \qquad
    {{\includegraphics[width=7.5cm]{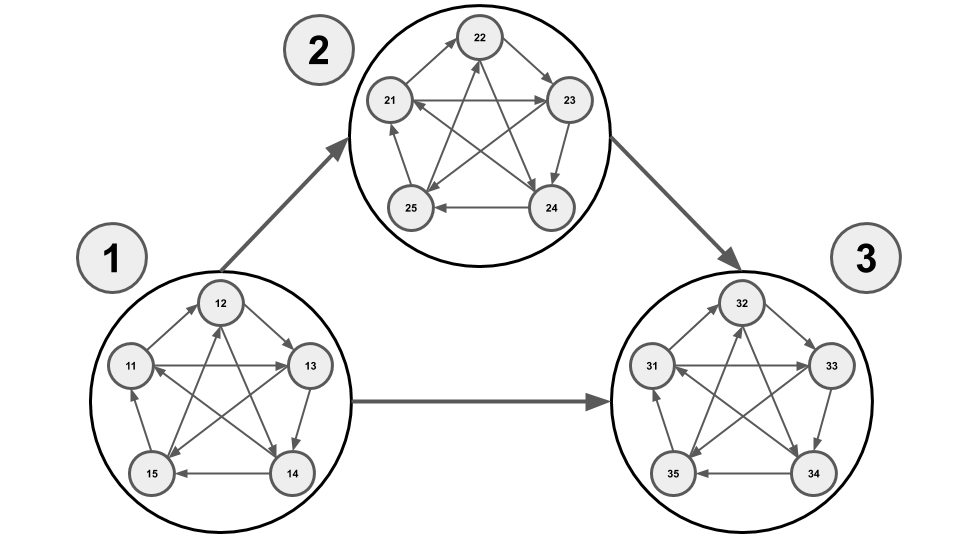} }}%
    \caption{A tournament $T$ (top) and its associated construction $T'$ (bottom), as described in the proof of Lemma~\ref{lem:almost-half} for $z=3$. The outcomes of the games between teams in different components of $T'$ mimic the outcomes of the games between the representative nodes in $T$.}%
    \label{fig:tournamentexpansion}%
\end{figure}

\begin{proof}
Consider an arbitrary tournament $T$ with $n$ teams, and consider a related tournament $T'$ with $n(2z-1)$ teams, labeled $v_{ij}$, for $i \in [n]$ and $j \in [2z-1]$. Conceptually, think that we have split each original team into a group of $2z - 1$ copies. For $i \ne j$, and $x,y \in [2z-1]$, have $v_{ix}$ beat $v_{jy}$ in $T'$ if and only if $v_{i}$ beat $v_{j}$ in $T$ (that is, match results in $T$ are preserved between different groups in $T'$). Within each group, have $v_{ix}$ beat $v_{iy}$ in $T'$ iff $x < y < x + z\pmod{2z-1}$ (so each group is isomorphic to the $z$-balanced tournament, recall Definition~\ref{def:balT} and see Figure~\ref{fig:tournamentexpansion} for a small example). Let $G(i)$ denote $i$'s group, $G(i):= \{v_{ij}, j \in [2z-1]\}$.

Now, we wish to claim that if $r(\cdot)$ is a rule that is $\infty$-$\SNM$-$\alpha$ (respectively, $k$-$\SNM$-$\alpha$) for $n(2z-1)$ teams, and we define $w(T)$ so that $w_i(T):= \sum_{j=1}^{2z-1} r_{ij}(T')$, then $w(\cdot)$ is a rule that is $\infty$-$\SNM$-$\alpha$ (respectively, $\frac{k}{2z-1}$-$\SNM$-$\alpha$) for $n$ teams, \emph{and $w(T)$ is a feasible solution to $LP_1(T, \alpha,z)$ for all $T$.}
	
Let's first confirm that $w(T)$ is a feasible solution to $LP_1(T, \alpha, z)$ with $\sum_i p_i=1$. The latter statement is clear: as $r(\cdot)$ is a tournament rule, we have $ \sum_i p_i = \sum_i w_i(T) = \sum_{ij} r_{i,j}(T) = 1$. Next, it is also clear that $p_i \geq 0$ for all $i$, so we just need to check that $\sum_{j \in \delta^-_T(i)} p_j + \frac{z}{2z-1}p_i \geq 1-\alpha$. 

To this end, we know that there exists \emph{some} adjacent set of $z$ teams in $G(i)$ such that the total probability that these teams win is at most $\frac{z}{2z-1}\cdot p_i$. Call this set $S_x$ and let $v_{ix}$ denote the team in this set which loses to the others. Then the set of teams $\cup_{j \in \delta^-_T(i)} G(j) \cup S_x$ together can create a Condorcet winner ($v_{ix}$) in $T'$. Therefore, we get that this set of teams must have won with probability at least $1-\alpha$ under $r(\cdot)$, and by definition of $w(\cdot)$ (and the choice of $S_x$ above), we immediately get that $\sum_{j \in \delta^-_T(i)}w_j(T) +\frac{z}{2z-1}w_i(T) \geq 1-\alpha$, as desired. 

So now we know that $w(\cdot)$ satisfies $LP_1(T,\alpha,z)$ with $\sum_i p_i = 1$ for all $T$. We now need to confirm that it is also $\infty$-$\SNM$-$\alpha$ (respectively, $\frac{k}{2z-1}$-$\SNM$-$\alpha$). But suppose for contradiction that $w(\cdot)$ was not $k$-$\SNM$-$\alpha$ for some $k$ (respectively, $\frac{k}{2z-1}$-$\SNM$-$\alpha$). This would imply the existence of tournaments $T_1$ and $T_2$ that are $S$-adjacent for some set $S \subseteq [n]$ (respectively, $S \subseteq [n]$, with $|S| \leq \frac{k}{2z-1}$) where $\sum_{i\in S} w_i(T_1) - \sum_{i \in S}{w_i(T_2)} > \alpha$. If we let $T'_1$ and $T'_2$ represent the corresponding tournaments that determined the values of $T_1$ and $T_2$ from $r(\cdot)$ respectively, and let $S' = \cup_{i \in S} G(i)$, we can conclude $\sum_{i\in S'} r_i(T'_1) - \sum_{i \in S'}{r_i(T'_2)} > \alpha$, contradicting the fact that $r$ is $\infty$-$\SNM$-$\alpha$ (respectively, that $r$ is $k$-$\SNM$-$\alpha$, as $|S'| = |S| \cdot (2z-1)$ and $|S| \leq \frac{k}{2z-1}$).
\end{proof}

With Lemma~\ref{lem:almost-half} in hand, we're very close to our goal. In particular, we've now shown that $\infty$-$\SNM$-$\alpha$ rules exist for all $n$ if and only if $\infty$-$\SNM$-$\alpha$ rules exist for all $n$ \emph{which additionally satisfy the constraints in $LP_1(T,\alpha,z)$ for all $z \in \mathbb{N}_+$}. Note that as $z \rightarrow \infty$, the constraints of $LP_1(T, 1/2, z)$ approach those of $SLP(T)$. So one might reasonably expect that $SLP(T)$ can be used in place of $LP_1(T,1/2,z)$ above, specifically when $\alpha = 1/2$. Indeed, this is the case (and the only place where we use Proposition~\ref{prop:epsilon-delta}). 
\begin{theorem}
	\label{thm:triangle}
	There exists an $\infty$-$\SNM$-$1/2$ tournament rule for all $n$ if and only if the $\SLP$ Tournament Rule is $\infty$-$\SNM$-$1/2$ for all $n$.

Moreover, if the $\SLP$ Tournament Rule is \emph{not} $\infty$-$\SNM$-$1/2$ for all $n$, there exists a pair of integers $k, n < \infty$ such that no $k$-$\SNM$-$1/2$ Tournament Rule exists on $n$ teams.
\end{theorem}
\begin{proof}
The proof follows from a proper application of Lemma~\ref{lem:almost-half} and Proposition~\ref{prop:epsilon-delta}. Suppose towards contradiction that the $\SLP$ Tournament Rule is not $k$-$\SNM$-$1/2$ for some $k, n$. This implies that there must be some tournaments $T,T' \in T_n$ and manipulating set $S$ which verify this fact by gaining probability $c > \frac{1}{2}$. Call $A(T), b(T)$ be the constraint matrix and vector of $SLP(T)$, respectively, when written in standard form (i.e. $A(T)$ has $n$ rows, corresponding to the $n$ non-trivial constraints in $SLP(T)$. $b(T)$ is just the $n$-dimensional vector of all $1/2$s).	 
	 
%	 \begin{equation}
%	 \begin{split}
%	 \text{maximize} ~ c \cdot x \\
%	 A\vec{x} \geq \vec{b} \\
%	 \vec{x} \geq \vec{0}
%	 \end{split}
%	 \end{equation}

Now apply Proposition~\ref{prop:epsilon-delta} with $A:=A(T)$ and $b:=b(T)$, with $\delta =\frac{c-\frac{1}{2}}{4}$, and let $\varepsilon(T)$ be the promised $\varepsilon$. Do the same for $T'$, and set $\varepsilon = \min\{\varepsilon(T), \varepsilon(T')\}$. Pick now a sufficiently large $z$ such that $\frac{z}{2z-1} - \frac{1}{2} \leq \varepsilon$ (such a $z$ exists as $\varepsilon > 0$).

Now, observe that any feasible solution $\vec{x}$ for $LP_1(T,1/2,z)$ satisfies $A(T)\cdot\vec{x}\geq \vec{b} - \varepsilon\vec{1}$ (and any feasible solution $\vec{y}$ for $LP_1(T',1/2,z)$ satisfies $A(T')\cdot \vec{y} \geq \vec{b}-\varepsilon\vec{1}$). If there is an $\infty$-$\SNM$-$1/2$ tournament rule (respectively, $k$-$\SNM$-$1/2$ tournament rule, for $k$ to be chosen later), Lemma~\ref{lem:almost-half} tells us that there exists an $\infty$-$\SNM$-$1/2$ (respectively, $\frac{k}{2z-1}$-$\SNM$-$1/2$) tournament rule $y$ such that $y(T)$ is feasible for $LP_1(T,1/2,z)$ and $y(T')$ is feasible for $LP_1(T',1/2,z)$. So we know $A(T)\cdot y(T)\geq \vec{b} - \varepsilon\cdot \vec{1}$, and also that $A(T') \cdot y(T') \geq \vec{b} - \varepsilon\cdot \vec{1}$. Proposition~\ref{prop:epsilon-delta} then allows us to conclude that $|y(T) - w(T)|_1 \leq \delta$, and also that $|y(T') - w(T')|_1 \leq \delta$. But now we are ready to derive a contradiction and claim that in fact $y(\cdot)$ is not $\infty$-$\SNM$-$1/2$. Indeed, we know that 
	 $$\sum_{v \in S} w_v(T) - w_v(T') \geq c,$$ 
by definition of $S, T, T'$. But from the triangle inequality, we get that:
$$\sum_{v \in S} w_v(T) - y_v(T) \leq \sum_{v \in S} |w_v(T) - y_v(T)| \leq |w(T) - y(T)|_1 \leq \delta,$$
$$\sum_{v \in S} y_v(T') - w_v(T') \leq \sum_{v \in S} |w_v(T') - y_v(T')| \leq |w(T') - y(T')|_1 \leq \delta.$$

Summing these three equations then yields:
	 $$\sum_{v \in S} y(T)_v - y(T')_v + 2 \delta \geq c \quad \Rightarrow \quad \sum_{v \in S} y(T)_v - y(T')_v \geq \frac{c+\frac{1}{2}}{2} > \frac{1}{2}.$$	
This contradicts that $y(\cdot)$ was $\infty$-$\SNM$-$1/2$ (and contradicts that $y(\cdot)$ is $\frac{k}{2z-1}$-$\SNM$-$1/2$, as long as $|S| \leq \frac{k}{2z-1}$, or $k \geq |S|(2z-1)$. Note that $k$ can indeed be defined after $z$ and $S$), as now $S$ can manipulate from $T'$ to $T$ and gain $>1/2$. 
\end{proof}

To briefly recap the entire proof of Theorem~\ref{thm:triangle}: we first showed that the existence of $\infty$-$\SNM$-$\alpha$ rules imply the existence of specific kinds of $\infty$-$\SNM$-$\alpha$ rules (those which satisfy $LP_1(T, \alpha, z)$ for all $z \in \mathbb{N}_+$). Note that we relied on the existence of $\infty$-$\SNM$-$\alpha$ rules for $n' \gg n$ in order to show the existence of our specialized $\infty$-$\SNM$-$\alpha$ rules for $n$. Then, we showed that for $\alpha = 1/2$, the existence of specialized rules implies that a particular rule (the $\SLP$ Tournament Rule) is $\infty$-$\SNM$-$1/2$ (and the fact that the $\SLP$ Tournament Rule is well-defined is the focus of Section~\ref{sec:LB}). 

Now, we make use of Theorem~\ref{thm:triangle} by proving that the $\SLP$ Tournament Rule is not $\infty$-$\SNM$-$1/2$.

\begin{lemma}
	\label{lem:counter-example}
The $\SLP$ Tournament Rule is not $\infty$-$\SNM$-$1/2$.
\end{lemma}
\begin{proof}
See Figure~\ref{fig:SLPgain} where two $\{B,C,E\}$-adjacent tournaments are evaluated under the $\SLP$ Tournament Rule. The three teams $\{B, C, E\}$ together have probability $4/9$ under $T$, but $1$ under $T'$, and therefore gain $5/9 > 1/2$ by manipulating. So the rule is not $\infty$-$\SNM$-$1/2$. 

Note that in order to verify that we have computed the $\SLP$ Tournament Rule correctly on $T$ and $T'$, the reader need only verify (in each graph) that the probabilities sum to $1$, and the $\SLP$ constraints: for all $i$, $\sum_{j \in \delta^-_T(i)} p_j + p_i/2 \geq 1/2$. By Proposition~\ref{prop:uniqueness}, any such solution is the unique optimum, and therefore output by the $\SLP$ Tournament Rule.

\begin{figure}
\centering
\includegraphics[width=0.75\textwidth]{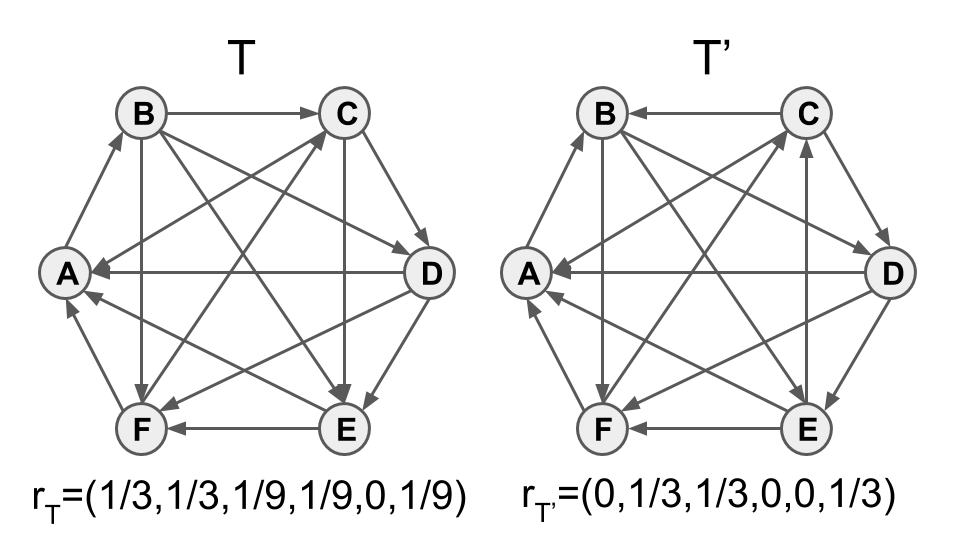}
\caption{The tournaments $T, T'$ are $\{B, C, E \}$ adjacent. The distribution of winners as prescribed by $\SLP$ is provided below each tournament, where the teams are presented alphabetically. Note that under $T$, the collusion wins with chance $4/9$ but in $T'$ they win with probability $1$, gaining $5/9$.}
\label{fig:SLPgain}
\end{figure}
\end{proof}

\begin{proof}[Proof of Theorem~\ref{thm:newLB}]
The proof of Theorem~\ref{thm:newLB} now follows immediately from Theorem~\ref{thm:triangle} and Lemma~\ref{lem:counter-example}. 
\end{proof}

\subsection{Concluding Thoughts on Lower Bounds}
We emphasize again that our lower bounds reduce the problem of determining existence of Condorcet-consistent $\infty$-$\SNM$-$1/2$ tournament rules for large $n' \gg n$ to the problem of determining whether the specific $\SLP$ Tournament Rule is $\infty$-$\SNM$-$1/2$ for small $n$. In particular, now that we have a specific $n, k$ for which the $\SLP$ Tournament Rule is not $k$-$\SNM$-$1/2$ on $n$ teams, we can backtrack through Theorem~\ref{thm:triangle} and recover a specific $k',n'$ for which no $k'$-$\SNM$-$1/2$ (and therefore no $k'$-$\SNM$-$\frac{k'-1}{2k'-1}$) Tournament Rule exists on $n'$ teams:

In our example, there are $n = 6$ teams, and the $\SLP$ Tournament Rule is not better than $3$-$\SNM$-$5/9$. So we may take $c = 5/9$ in the proof of Theorem~\ref{thm:triangle}, which results in $\delta = 1/72$. Note that $\varepsilon$ is now a function of $\delta$ via Proposition~\ref{prop:epsilon-delta}, and is $\approx 0.0016$, so we'd set $z:= \lceil \frac{1/2+\varepsilon}{2\varepsilon}\rceil = 157$. This therefore rules out the possibility of a tournament rule that is $939$-$\SNM$-$1/2$ for $1878$ teams. While of course portions of the proof of Theorem~\ref{thm:triangle} could be optimized to yield a smaller $k', n'$, the point we are trying to make is that there could very well be $k$-$\SNM$-$1/2$ tournament rules for $n$ teams for quite large values of $k, n$, and there is virtually no hope of uncovering the non-existence for extremely large $k$ via exhaustive search --- recall that the space of tournament rules is all functions from the $2^{\binom{n}{2}}$ different complete directed graphs on $n$ teams to the $n$-dimensional simplex (indeed, the authors had no luck via exhaustive search, and numerous rules appeared $k$-SNM-$1/2$ for small $n$ in simulations). However, the machinery developed in this section allows us to brute-force search for manipulations of a \emph{single} tournament, which happened to resolve for $k=3, n=6$ and conclude our desired claim, which would have required a significantly larger exhaust for \emph{significantly} larger $k', n'$. 

Finally, we note that it may be tempting to use our machinery, almost as is, to rule out Condorcet-consistent $\infty$-$\SNM$-$\alpha$ rules for $\alpha > 1/2$. In particular, it is tempting to conclude that because Figure~\ref{fig:SLPgain} exhibits that the $\SLP$ Tournament Rule is no better than $3$-$\SNM$-$5/9$, that there should not be an $\infty$-$\SNM$-$5/9$ tournament rule for all $n$. Note, however, that the $\SLP$ is really special for $\alpha = 1/2$. Indeed, if we were to replace $1/2$ with $4/9$ in the SLP, we would (for instance) no longer have a unique solution. Therefore, we'd lose the well-defined $\SLP$ Tournament Rule, and still have to do a broad exhaustive search, and significantly new ideas would be needed to get leverage out of this. Still, while our current tools only preclude rules which are $\infty$-SNM-$1/2$ (and ever so slightly more: $\infty$-SNM-$.50016$ via similar reasoning to the previous paragraph), it is reasonable to expect that our general approach (e.g. Lemma~\ref{lem:almost-half}) may help rule out the existence of SNM-$\alpha$ tournament rules for $\alpha > 1/2$.

\section{Less Manipulable Tournament Rules via $\SLP$}
\label{sec:UB}

In this section we prove Theorem~\ref{thm:infUB}: an $\infty$-$\SNM$-$2/3$ tournament rule exists (for all $n$). Fortunately, a lot of the work has been done in Sections~\ref{sec:LP} and~\ref{sec:LB} in the form of our understanding of how feasible solutions to $SLP(T)$ and $LP_0(T,\beta,k)$ relate. We first show how any $k$-$\SNM$-$\alpha$ rule that is a valid solution to $LP_0(T,\beta,k)$ (for some $\beta$) can be transformed into a $k$-$\SNM$-$\frac{\alpha}{\alpha-\beta+1}$ rule. This transformation yields a stronger tournament rule if $\alpha \geq \beta$. Our proof then exploits the fact that an optimal solution to $SLP(T)$ satisfies $LP_0(T,1/2,k)$ by design and, trivially, is $\infty$-$\SNM$-$1$. The naive upper bound of $\alpha \leq 1$ suffices to yield a $\infty$-$\SNM$-$2/3$ rule (previously no $\infty$-$\SNM$-$<1$ rule is known). Moreover, this reduction now allows any improved bounds (even if $\gg 2/3$) on the manipulability of the $\SLP$ Tournament Rule to imply tournament rules which are $\infty$-$\SNM$-$\alpha$ for $\alpha < 2/3$. Below is the main lemma of this section.

\begin{lemma}[Augmentation Lemma]
\label{lem:aug}
	Let tournament rule $r(\cdot)$ be such that $r(T)$ satisfies $LP_0(T, \beta, k)$ for all tournaments $T$, and be $k$-$\SNM$-$\alpha$. Then a tournament rule $w(\cdot)$ exists which is $k$-$\SNM$-$\frac{\alpha}{\alpha - \beta + 1}$.
\end{lemma}
\begin{proof}
	Consider the following rule: pick a $c \in [0,1]$ (to be chosen later). Set $w_i(T) = r_i(T)\cdot c + \frac{1-c}{n}$, if $T$ does not have a Condorcet winner. If $T$ has a Condorcet winner, allocate probability $1$ to the Condorcet winner. 

To evaluate the manipulability of this rule, first consider a manipulating set which creates a Condorcet winner. The total probability gained is at most $\beta \cdot c + (1-c)$. To see this, observe that because the set can create a Condorcet winner, they must have total probability at least $\beta$ under $r(\cdot)$ (and therefore at least $\beta \cdot c$ after scaling down by $c$). 

Now consider a manipulating set that does not create a Condorcet winner. Then there is certainly no Condorcet winner in $T'$, and so the extra $(1-c)$ probability mass is still allocated uniformly (and the set gains nothing here). So the set can only gain what they would by manipulating under $r(\cdot)$ (scaled down by $c$), which is at most $\alpha \cdot c$. 

To minimize $\max\{\alpha c, \beta c +(1-c)\}$, set $c:= \frac{1}{\alpha - \beta + 1}$. This results in $w(\cdot)$ being $k$-$\SNM$-$\frac{\alpha}{\alpha - \beta + 1}$.
\end{proof}

\begin{proof}[Proof of Theorem~\ref{thm:infUB}]
By construction the $\SLP$ Tournament Rule is Condorcet-consistent, and is feasible for $LP_0(T,\frac{1}{2}, k)$ for all $k$. Thus the $\SLP$ Tournament Rule satisfies the requirements of the augmentation lemma for $\beta = \frac{1}{2}$ and $\alpha = 1$ (as all rules are $k$-$\SNM$-$1$ for all $k$) for all $k$, so Lemma~\ref{lem:aug} results in an $\infty$-$\SNM$-$2/3$ rule.
\end{proof}

At this point the experienced reader may wonder about other useful properties of the $\SLP$ Tournament Rule. In Appendix~\ref{ubapp} we show that the $\SLP$ Tournament Rule is \emph{not} monotone. 

\section{Cover-Consistent Tournament Rules}\label{sec:cover}
In this section we shift gears and return to $2$-$\SNM$-$\alpha$ tournaments. We extend the results of~\cite{SchneiderSW17} not in the direction of larger $k$ or smaller $\alpha$, but towards a more stringent requirement than Condorcet-consistence (cover-consistence). The main result of this section is Theorem~\ref{thm:koth}, which develops a new tournament rule which is cover-consistent and $2$-$\SNM$-$1/3$ (the smallest $\alpha$ possible, by Lemma~\ref{lem:sswlb}). We call our rule Randomized-King-of the Hill and define it below. %\albnote{I didn't think too hard about the name when I was initially writing this up, so if anyone has a better name for it I'm all ears.}\mattnote{I like the name!}

\begin{definition}[Randomized-King-of-the-Hill] The \emph{Randomized-King-of-the-Hill} Tournament Rule (\RKOTH) starts every step by first checking whether there is a Condorcet winner among the remaining teams. If so, that team is declared the winner. If not, it picks a uniformly random remaining team $i$ (which we'll call the \emph{prince}) and removes team $i$ and all teams which lose to $i$. Algorithm~\ref{alg:2snm} provides pseudocode. 
%\footnote{Note that in Algorithm~\ref{alg:2snm} when we update $V$, we also remove the team picked as prince-of-the-hill. This is without loss of generality since once a team becomes randomly selected by the algorithm, it will get eliminated on the next round no matter what. Hence, the pseudocode is indeed equivalent to the text description here.} %\albnote{One other minor detail is that in the text description, the algorithm checks if there is a Condorcet winner among all remaining teams, whereas in the pseudocode, it only checks whether the randomly selected team is a Condorcet winner (similar to the lazy/eager difference). In the end it is still equivalent, since the Condorcet winner will never be eliminated. In terms of runtime, my intuition is that the lazy version is slightly better; at each step, the expected number of teams eliminated is $\frac{n-1}{2}$, so expected runtime for lazy should be around $\log n$, whereas eager is at least $n$.}\mattnote{I think it's OK to not stress about this distinction. We probably could get away without the existing footnote too, but we're not tight on space.}}
\end{definition}   

\begin{algorithm}[t]
	\caption{Pseudocode for the Randomized-King-of-the-Hill Tournament Rule.}
	\label{alg:2snm}
	\SetAlgoNoLine
	\KwIn{A tournament graph $T = (V, E)$ on $n$ teams.}
	\KwOut{A winning team $i \in V$.}
	\Repeat{}{
	Choose a team $j \in V$ uniformly at random \; 
		\If{($j$ is a Condorcet winner)\;}
		{return $j$\;}
	{}{
		$V \leftarrow V \setminus \{j \cup \delta^+_T(j)\}$\; 
		}
		}
\end{algorithm}

\iffalse
\arnote{this text may be useful in other parts of paper (i.e. intro) if you think its rambling here}\mattnote{I agree the intro is better. Iffalsed it for now.}

Algorithm~\ref{alg:2snm} queries one edge from every remaining non-prince-of-the-hill team per turn and all the edges from the prince-of-the-hill. In contrast, the randomized single-bracket elimination tournament rule proposed in~\cite{SchneiderSW17} queries exactly $n-1$ edges in the graph in total. Hence our proposed rule always queries at least as many edges as the single-bracket elimination rule. One can argue that since it gathers more information from $T$, Algorithm~\ref{alg:2snm} makes "fairer" decisions than the single-bracket elimination rule. 

\fi
The main distinction we'll emphasize between \RSEB\ and \RKOTH\ is that \RKOTH\ is cover-consistent (Lemma~\ref{lem:cc} below), while \RSEB\ is not (Observation~\ref{obs:RSEBcover}). We later show that \RKOTH\ is also $2$-$\SNM$-$1/3$, just like \RSEB.

\begin{lemma}\label{lem:cc} \RKOTH\ is cover-consistent.
\end{lemma}
\begin{proof}
Consider any two teams $u,v$ where $u$ covers $v$. If \RKOTH\ is to possibly output $v$, the team $u$ must be removed at some round where team $x$ is selected. If at the start of this round, $v$ has already been removed, then $v$ will clearly not be declared the winner. If at the start of this round, $v$ has not already been removed, then $v$ is removed this round because $v$ loses to $x$ (as $x$ beats $u$ and $u$ covers $v$). Therefore, $v$ can never be declared the winner by \RKOTH. 
\end{proof}

In fact, $\RKOTH$ satisfies an even stronger property than cover-consistency. Before stating this we need to introduce some definitions.

\begin{definition}[Sub tournament] 
A \emph{sub tournament} of a tournament $T$ with respect to a set of teams $S$ is the tournament induced by the games between teams in $S$.  
\end{definition} 

\begin{definition}[Transitive tournaments]
A tournament $T$ is \emph{transitive} if there are no directed cycles. 
\end{definition}

\begin{definition}[Banks set]
Team $v$ is a Banks winner of tournament $T$ if there exists a maximal (with respect to inclusion) transitive sub tournament $T'$ of $T$ where $v$ is the Condorcet-winner. The \emph{Banks set} of a tournament $T$ is the set of Banks winners of the tournament.
\end{definition}

\begin{claim}
\label{claim:banks}
The Banks set of a tournament is a subset of the set of uncovered teams of the tournament. 
\end{claim}

\begin{proof}
We will show the contrapositive. Consider a team $v$ that is covered by some other team $u$. No maximal transitive sub tournament of $T$ can have $v$ as its Condorcet-winner because $u$ beats $v$ and everyone $v$ beats and hence can always be added on top of $v$, contradicting the maximality of the sub tournament.  
\end{proof} 

We can now prove Lemma~\ref{lem:bankskoth}, which states that the Banks set of a tournament $T$ is \emph{exactly} the set of teams \RKOTH\ can declare as winner. 

\begin{proof}[Proof of Lemma~\ref{lem:bankskoth}]

First, we argue that any Banks winner can be output by $\RKOTH$ with non-zero probability. Indeed, consider any Banks winner $v$, and let $T'$ denote the maximal transitive subtournament in which $v$ is the Condorcet winner. Name the teams in $T'\setminus \{v\}$ as $u_1,\ldots, u_k$, where $u_i$ beats $u_j$ for all $j < i$ (and $v$ beats $u_i$ for all $i$). Now, because $T'$ is a maximal transitive subtournament, there does not exist any $w$ which beats all teams in $T'$ (otherwise we could add $w$ to the subtournament, witnessing non-maximality). Now, consider an execution of $\RKOTH$ which first selects princes in order of $u_i$ ($u_1$, then $u_2$, etc.) and then finally $v$. First, observe that each $u_i$ has not yet been eliminated by the time we hope they are selected (by definition of $u_i$). Second, observe that every team $w \notin T'$ must be eliminated by the end, because they lose to some team in $T'$. Therefore, after this execution, $v$ is the only remaining team, and crowned champion.

It remains to show that any team in the support of $\RKOTH$ is always a Banks winner. Out of all the executions of $\RKOTH$ where $v$ wins, consider one that goes through the most princes before $v$ is picked. We claim this will be a maximal transitive sub tournament. Let $P = \{p_1,..., p_k\}$ be the set of princes used by the algorithm in that order. Since $v$ wins under those princes, it must beat every team in $P$. Moreover, since the algorithm first picks $p_1$, then $p_2$, and so on, it must be the case that $p_i$ beats $p_j$ for $i > j$. Otherwise $p_i$ would be eliminated on the step where $p_j$ was selected as the prince. Therefore, the sub tournament $v \cup P$ is transitive, and has $v$ as its Condorcet-winner. 

Consider any team $x \not \in v \cup P$. If $x$ beat teams $p_1,..., p_i$ and lost to teams $p_{i+1},..., v$ for some $i \in \{0,...,k+1\}$, then an execution that places $x$ between $p_i$, $p_{i+1}$ would be feasible. It would still output $v$ but run for one more step than $P$, contradicting the maximality of $P$. Therefore, $v \cup P$ is a maximal transitive sub tournament where $v$ is a Condorcet-winner, implying $v$ is in the Banks set. 

\end{proof}

We will also use the fact that \RKOTH\ is monotone in our remaining proof. Below (and for the remainder of this section), we'll refer to a \emph{prince} as the most recently selected remaining team, and we'll refer to an \emph{execution} of \RKOTH\ as simply an ordering over potential princes (to be selected if they haven't yet been eliminated when their turn comes).% \albnote{Should we move the prince naming convention to right after the first paragraph, since footnote 5 uses it?}\mattnote{Yes, good catch.}

\begin{lemma}\label{lem:monotone} \RKOTH\ is monotone. That is, if $T, T'$ are $\{u,v\}$-adjacent and $u$ beats $v$ in $T$, then $r_u(T) \geq r_u(T')$.
\end{lemma}
\begin{proof}
Consider any execution of \RKOTH\, and consider the first time that either $u$ or $v$ is prince (observe that prior to this, the edge between $u$ and $v$ is never queried, so the execution on $T$ and $T'$ is identical). If either $u$ or $v$ is already eliminated, then it doesn't matter whether $u$ beats $v$ or vice versa, and the outcome is the same. Otherwise, if $u$ is the prince and $u$ beats $v$, then there is a chance that $u$ wins. If $u$ loses to $v$, then $u$ is eliminated immediately. If $v$ is the prince and $u$ beats $v$, then there is a chance that $u$ wins. If $u$ loses to $v$, then $u$ is eliminated immediately. Therefore, for every execution, if $u$ wins in $T'$, $u$ also wins in $T$, and the lemma holds.
\end{proof}

The rest of this section is devoted to proving that \RKOTH\ is $2$-$\SNM$-$1/3$. The key approach of our proof is the following: consider any round in which both $u$ and $v$ still remain. The next team selected as prince might beat both $u$ and $v$ (in which case the outcome between $u$ and $v$ is never queried), lose to both $u$ and $v$ (in which case the outcome has not yet been queried), or beat exactly one of $\{u,v\}$ (in which case again the outcome between $u$ and $v$ is never queried). The only event in which we ever query the outcome of the match is when one of $\{u,v\}$ is selected as prince while the other remains (and even then, the outcome only matters if some teams remain which beat exactly one of $\{u,v\}$). So the key approach in the proof is a coupling argument between different possible executions of \RKOTH\ (some of which make the match between $u$ and $v$ irrelevant, and some of which cause $\{u,v\}$ to prefer the match turn one way or the other). 

\begin{proof}[Proof of Theorem~\ref{thm:koth}]
In order to show the rule is $2$-$\SNM$-$1/3$, consider two teams $u, v$ who are trying to collude in a given tournament $T$. Suppose wlog that $u$ beats $v$ in $T$ and let $T'$ be the $\{u,v\}$-adjacent tournament to $T$ where $v$ beats $u$. 

To begin the analysis, we first introduce some notation. Let $S$ be the subset of teams which either beat at least one of $u$ or $v$, or are $u$ or $v$. For a given execution of \RKOTH\, let $x$ denote the first prince in $S$ on tournament $T$. Observe first that there must be a prince in $S$ at some point (otherwise neither $u$ nor $v$ is ever eliminated), and also that $x$ is the first prince in $S$ on tournament $T'$ as well (for the same execution). Let also $X$ denote the set of princes \emph{strictly before} $x$ was chosen, and $Y(X)$ denote the set of un-eliminated teams after the set $X$ of princes. We first observe that, conditioned on $X$, the next prince is a uniformly random element of $Y(X) \cap S$. 

\begin{lemma}\label{lem:first} For all $X \subset [n] \setminus S$, conditioned on the set $X$ being princes so far, and the next prince being an element of $S$, the next prince is a uniformly random element of $Y(X) \cap S$.
\end{lemma}
\begin{proof}
For all $X \subseteq [n] \setminus S$, conditioned on the set $X$ being princes so far, the next prince is a uniformly random element of $Y(X)$, so each element of $Y(X) \cap S$ is selected with equal probability.
\end{proof}

The main step in the proof is the following lemma, which claims that after conditioning on $X$, the difference between $T$ and $T'$ in terms of whether one of $\{u,v\}$ wins under \RKOTH\ is small.

\begin{lemma}\label{lem:final} For all $X$, let $r_{u,v}(T,X)$ denote the probability that \RKOTH\ selects a winner in $\{u,v\}$, conditioned on $X$ being exactly the set of princes before the first prince in $S$. Then for all $X$, $|r_{u,v}(T',X)-r_{u,v}(T,X)| \leq 1/3$. 
\end{lemma}
\begin{proof}
We consider a few possible cases, conditioned on the structure of $Y(X) \cap S$, and which team is the next prince. To aid in formality, we'll use the notation $r_{u,v}(T,X|E)$ to denote the probability that \RKOTH\ selects a winner in $\{u,v\}$ on tournament $T$ conditioned on exactly the set $X$ of princes before the first prince in $S$ and event $E$.\\

\noindent\textbf{Case One: $Y(X)\cap S = \{u,v\}$.} In this case, certainly $u$ or $v$ will win in tournament $T$ and $T'$. This is because all remaining teams lose to both $u$ and $v$, so whoever wins the match between $u$ and $v$ is a Condorcet winner among the remaining teams and will therefore win. So if $E_1$ denotes the event that $Y(X) \cap S = \{u,v\}$, we have that $r_{u,v}(T,X|E_1) = r_{u,v}(T',X|E_1)$.\\

\noindent\textbf{Case Two: $|Y(X)\cap S| > 2$, next prince $\notin\{u,v\}$.} In this case, at least one of $\{u,v\}$ are eliminated immediately, and the match result is never queried. Therefore, the result is the same under $T$ and $T'$. So if $E_2$ denotes the event that $|Y(X) \cap S| > 2$ and the next prince is not in $\{u,v\}$, we have that $r_{u,v}(T,X|E_2) = r_{u,v}(T',X|E_2)$. \\

\noindent\textbf{Case Three: $|Y(X) \cap S| > 2$, next prince is $u$.} In this case, we claim it is \emph{always better} for $\{u,v\}$ to be in $T'$ ($v$ beats $u$) versus $T$. To see this, first consider that maybe some remaining element of $Y(X) \cap S$ beats $u$. Then $u$ certainly will not win. If $u$ beats $v$, then $v$ also will certainly not win. But if $v$ beats $u$, then maybe $v$ can win. If no remaining element of $Y(X) \cap S$ beats $u$, then either $u$ will win in $T$, or $v$ will win in $T'$ (because all teams aside from $u$ and $v$ are eliminated). So if $E_3$ denotes the event that $|Y(X) \cap S| > 2$ and the next prince is $u$, we have that $r_{u,v}(T',X|E_3) \geq r_{u,v}(T,X|E_3)$. \\

\noindent\textbf{Case Four: $|Y(X) \cap S| > 2$, next prince is $v$.} This case is symmetric to the above, and it is always better for $\{u,v\}$ to be in $T$ versus $T'$. So if $E_4$ denotes the event that $|Y(X) \cap S| > 2$ and the next prince is $v$, we have that: $r_{u,v}(T,X|E_4) \geq r_{u,v}(T',X|E_4)$.\\

So to conclude, we've seen that if $|Y(X)\cap S| = 2$, then the outcome is the same under $T$ and $T'$, so the lemma statement clearly holds for any $X$ with $|Y(X) \cap S| = 2$. If $|Y(X)\cap S| > 2$, then there is \emph{exactly} one choice for the next team which \emph{may} cause $\{u,v\}$ to prefer $T$ to $T'$ (and vice versa). By Lemma~\ref{lem:first}, the next prince is drawn uniformly at random from $|Y(X)\cap S|$, so this element is selected with probability at most $1/3$. Formally, for any $X$ with $|Y(X)\cap S| > 2$ we have:
\begin{align*}
r_{u,v}(T,X) -r_{u,v}(T',X) =& \Pr[E_2]\cdot (r_{u,v}(T,X|E_2) - r_{u,v}(T',X|E_2)) \\
&+ \Pr[E_3] \cdot (r_{u,v}(T,X|E_3) - r_{u,v}(T',X|E_3))\\
 &+ \Pr[E_4] \cdot (r_{u,v}(T,X|E_4) - r_{u,v}(T',X|E_4))\\
\leq& \Pr[E_2] \cdot 0 + \Pr[E_3] \cdot 0 + \Pr[E_4] \cdot 1 \leq \frac{1}{|Y(X) \cap S|} \leq 1/3.
\end{align*}
Similar inequalities hold for $r_{u,v}(T',X)-r_{u,v}(T,X)$ but with the role of $E_3$ and $E_4$ swapped, allowing us to conclude that indeed $|r_{u,v}(T,X) - r_{u,v}(T',X)| \leq 1/3$.
\end{proof}

The rest of the proof now follows easily. Below, if $r_{u,v}(T)$ denotes the probability that either $u$ or $v$ wins under \RKOTH\ for tournament $T$, and $p(X)$ denotes the probability that $X$ is exactly the set of princes before the first prince in $S$ is selected (under tournament $T$), we have:

$$|r_{u,v}(T) - r_{u,v}(T')| \leq \sum_X p(X) \cdot |r_{u,v}(T,X) - r_{u,v}(T',X)| \leq 1/3.$$
\end{proof}

While \RKOTH\ and \RSEB\ are optimal tournament rules against collusions of size $2$, their effectiveness disappears as the collusions become larger. In particular, in Appendix~\ref{ubapp} we show a specific tournament against both rules for which large collusions can increase their odds of winning up to close to $1$.

\section{Conclusion}
\label{sec:conclusion}
We extend work of~\cite{SchneiderSW17} in three different directions: First, we refute their main conjecture (Theorem~\ref{thm:newLB}, Sections~\ref{sec:LP} and~\ref{sec:LB}). Next, we design the first Condorcet-consistent tournament rule which is $\infty$-SNM-$(<1)$ (Theorem~\ref{thm:infUB}, Sections~\ref{sec:LP} and~\ref{sec:UB}). Finally, we design a new tournament rule (\RKOTH) which is $2$-SNM-$1/3$ (just like \RSEB), but which is also cover-consistent (Theorem~\ref{thm:koth}, Section~\ref{sec:cover}). 

Reiterating from Section~\ref{sec:intro}, the main appeal of our results is clearly theoretical, and some of this appeal comes from the process itself. For example, Theorem~\ref{thm:triangle} reduces the search for a Condorcet-consistent $\infty$-SNM-$1/2$ rule to determining whether or not the \SLP\ Tournament Rule is $\infty$-SNM-$1/2$. Additionally, the same tools developed in Section~\ref{sec:LP} proved useful both for proving lower bounds and designing new tournament rules, suggesting that these tools should be useful in future works as well.

One clear direction for future work, now that the main conjecture of~\cite{SchneiderSW17} is refuted, is to understand what the minimum $\alpha$ is such that an $\infty$-SNM-$\alpha$ tournament rule exists. It is also interesting to understand how large $k$ needs to be in order for the~\cite{SchneiderSW17} conjecture to be false. Our work does not rule out the existence of a $3$-SNM-$2/5$ tournament rule, yet we also do not know of any $3$-SNM-$2/5$ rule (nor even a $3$-SNM-$1/2$ rule). More generally, our work contributes to the broad agenda of understanding the tradeoffs between incentive compatibility and quality of winner selected in tournament rules, and there are many interesting problems in this direction.

\section*{Acknowledgement}
The authors are extremely grateful to Mikhail Khodak and Jon Schneider, who contributed both with many helpful discussions as well as code to help test the non-manipulability of tournament rules. The authors would also like to thank the anonymous reviewers for their feedback on extensions, clarifications and relevant references unknown to the authors.

% Appendix
% \nocite{*}

\bibliographystyle{acm}
\bibliography{bibliography}

\appendix
% !TeX root = main.tex
% !TEX root = main.tex

\section{Missing proofs from Section~\ref{sec:linalg}}
\label{app:prelim}
To prove Proposition~\ref{prop:skew-symmetric-rank}, we first need some additional machinery. 

\begin{definition}[Pfaffian] The \emph{Pfaffian} of an $n\times n$ skew-symmetric matrix $A$ (when $n=2k$ is even) is defined as follows. First, let $\Pi$ be the set of all partitions of $[2k]$ into pairs without regard to order. If we write an element $\alpha \in \Pi$ as $\{(i_1,j_1),(i_2,j_2), \cdots (i_k,j_k)\}$ with $i_\ell < j_\ell$ for all $\ell$ and $i_1 < i_2 < \ldots < i_k$, then let $\pi_\alpha$ denote the permutation with $\pi_\alpha(2\ell-1) = i_\ell$ and $\pi_\alpha(2\ell) = j_\ell$ (i.e. $\pi_\alpha$ sorts elements in the order $i_1, j_1, i_2, j_2,\ldots$). Let $A_\alpha := sgn(\pi_\alpha)\prod_{\ell=1}^k A_{i_\ell,j_\ell}$.\footnote{Recall that the sign of a permutation $\pi$, denoted here by $sgn(\pi)$ is equal to the number of inversions in $\pi$ (that is, the number of pairs $i<j$ with $\pi(i) > \pi(j)$).} Then the Pfaffian of $A$, denoted by $Pf(A)$, is equal to $\sum_{\alpha \in \Pi} A_\alpha$.
\end{definition}

\begin{theorem}\label{thm:Cayley} When $n$ is even, any $n\times n$ skew-symmetric matrix $A$ satisfies $Det(A) = (Pf(A))^2$.
\end{theorem}

\begin{proof}[Proof of Proposition~\ref{prop:skew-symmetric-rank}]
We will first prove the case when $n$ is even, using Theorem~\ref{thm:Cayley}. We first observe that as $A$ is unit skew-symmetric, $|A_{ij}| = 1$ for all $i \neq j$. Therefore, for all $\alpha \in \Pi$, $|A_\alpha| = 1$. So the Pfaffian of $A$ is the sum of $|\Pi|$ terms, each of which are $\pm 1$. Therefore, if we can show that $|\Pi|$ is odd, the Pfaffian must be non-zero, Theorem~\ref{thm:Cayley} implies that $Det(A) \neq 0$ as well (meaning that $A$ has full rank). 

It is straight forward to count the number of elements in $\Pi$: first, there are $n-1$ choices for the partner of $1$. Then, there are $n-3$ choices for the partner of the smallest unpartnered element. After choosing the first $\ell$ pairs, there are $n-2\ell -1$ choices for the partner of the smallest unpartnered element. So there are $(n-1)\cdot (n-3) \ldots 5 \cdot 3$ elements of $\Pi$, which is an odd number. 

Now we consider the case where $n$ is odd. Observe first that the submatrix of $A$ created by removing the last row and last column is a unit skew symmetric matrix with even dimension, and therefore has rank $n-1$. Therefore, the rank of $A$ is also at least $n-1$. By Jacobi's theorem, all skew symmetric matrices of odd dimension have determinant zero (and are therefore not of full rank). Therefore, the rank must be exactly $n-1$.
\end{proof}

\begin{proof}[Proof of Proposition~\ref{prop:epsilon-delta}]
For a given $\delta > 0$, let $Y$ be the set of all vectors in the unit hypercube with $\ell_1$ distance strictly less than $\delta$ from some element of $P_0$ (i.e. there exists an element in $P_0$ within $\ell_1$ distance strictly less than $\delta$). It is easy to see that $Y$ is an open set. Now let $S=[0,1]^{n} \setminus Y$. Note that $S$ is closed and bounded, and therefore compact. 

Definte now the function $f(\vec{y}) = \max \{0,\max_{i \in [n]} \vec{b}_i - (A\cdot \vec{y})_i \}$. First observe that $f(\cdot)$ is continuous, as it is a composition (via maximums, subtractions, etc.) of affine functions. Observe also that $\vec{y} \in P_{\varepsilon}$ for all $\varepsilon \geq f(\vec{y})$, and that $\vec{y} \notin P_{\varepsilon}$ for all $\varepsilon < f(\vec{y})$. We now consider two possible cases for $\inf_{\vec{y} \in S} \{f(\vec{y})\}$, and find our desired $\varepsilon$.
	
	\begin{itemize}
	\item Case 1: $\inf_{\vec{y} \in S}\{f(\vec{y})\}= 0$. As $S$ is compact, $f$ achieves its infimum over $S$. Therefore, there exists a $\vec{y} \in S$ with $f(\vec{y}) = 0$. Let's parse what this means. First, as $f(\vec{y}) = 0$, we know that $\vec{y} \in P_0$. But also, $P_0 \subseteq Y$, and $S \cap Y = \emptyset$. So any $\vec{y} \in S$ cannot also be in $P_0$, which means $f(\vec{y}) > 0$ for all $\vec{y} \in S$, meaning that we can't have $\inf_{\vec{y} \in S} \{f(\vec{y})\} = 0$ after all.
	\item Case 2: $\inf_{\vec{y} \in S}\{f(\vec{y})\} = c > 0$. Then let $\varepsilon = c/2$. Now, observe that all elements of $S$ are \emph{not} in $P_{\varepsilon}$, as they all have $f(\vec{y}) > \varepsilon$. 
	\end{itemize}

So to wrap up, we must have $\inf_{\vec{y} \in S} \{f(\vec{y})\} = c > 0$, and if we set $\varepsilon = c/2$, then $S \cap P_{\varepsilon} = \emptyset$. Therefore, as $S \cap Y = [0,1]^n$, it must be the case that all of $P_{\varepsilon}$ is contained in $Y$. But this is exactly the desired statement: all elements of $P_{\varepsilon}$ have some point in $P_0$ within distance $\delta$. So take this to be our desired $\varepsilon$.
\end{proof}

\section{Further properties of the proposed tournament rules}
\label{ubapp} 

In this section we analyze further properties of the two main tournament rules discussed on this paper. First we show that the $\SLP$ tournament rule fails to satisfy monotonicity, a very natural and desirable property from the point of view of a tournament designer. 

\begin{claim}
$\SLP$ is not a monotone tournament rule. 
\end{claim}

\begin{figure}
\centering
\includegraphics[width=1\textwidth]{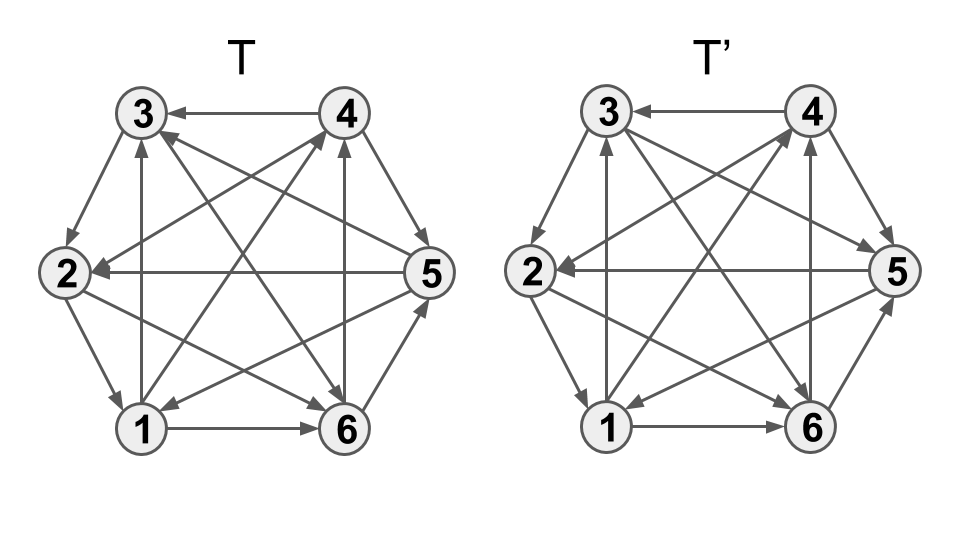}
\caption{The tournament $T$ can be unilaterally manipulated by team $5$ to become the tournament $T'$}
\label{fig:nonmonotone}
\end{figure}

\begin{proof}
Consider the tournament $T$ and it's $\{3,5\}$-adjacent tournament $T'$, both depicted in Figure~\ref{fig:nonmonotone}. In $T$, $5$ beats $3$ originally and the  $\SLP$ tournament rule awarded team $5$ a $.2$ chance of winning. If instead $5$ purposely throws its game to $3$, the $\SLP$ tournament rule \emph{rewards} team $5$ by increasing its chance of winning to $1/3$. Therefore the $\SLP$ tournament rule is not monotone.
\end{proof} 

Next we show that the optimality of $\RSEB, \RKOTH$ against collusions of size $2$ fails to translate to larger collusions. In particular, we show that the manipulability of both rules tends to $1$ as the size of the collusion increases, ruling both of them out as candidates for $\infty$-$\SNM$-$\alpha$ for constant $\alpha < 1$. First we define a family of tournaments that will be useful in showing lower bounds for both rules. These tournaments were introduced in~\cite{SchneiderSW17}.

\begin{definition}[Kryptonite Tournament~\cite{SchneiderSW17}]
A kryptonite tournament $T$ on $n$ teams has a \emph{superman} team (wlog label it $1$) who beats everyone except a \emph{kryptonite} team (wlog label it $n$). Moreover, team $n$ loses to every other team except $1$. The outcomes of the matches between the remaining teams may be arbitrary. 
\end{definition}

We now proceed with the main claim of this section. 

\begin{claim} 
$\RKOTH$ is no better than $k$-$\SNM$-$(k-1)/(k+1)$ for any $k$. $\RSEB$ is no better than $k$-$\SNM$-$(k-1)/k$ when $k+1$ is a power of $2$.
\end{claim}

\begin{proof}
Consider $k$ and let $n = k+1$. Let $T$ be any kryptonite tournament on $n$ teams, with team $1$ as the superman team and team $n$ as the kryptonite team. In order for team $1$ to win under $\RKOTH$, the rule must avoid selecting teams $1$ and $n$. Under any other first choice of prince, $\RKOTH$ will declare team $1$ the winner. Therefore it declares team $1$ the winner with probability $(k-1)/(k+1)$. However, if all teams but the superman collude they can make the kryptonite a Condorcet-winner. The collusion's combined odds of winning before were exactly $2/(k+1)$, so $\RKOTH$ is at least $k$-$\SNM$-$(k-1)/(k+1)$. 

Now let $k+1$ be a power of two and consider the same $T$ as before. In order for the superman team to win under $\RSEB$ it must avoid the kryptonite team in the first round. Therefore $\RSEB$ crowns the superman team winner with probability $(k-1)/k$. However, all other teams can form a collusion and turn the kryptonite team into a Condorcet-winner, increasing their winning mass by $(k-1)/k$. Therefore, $\RSEB$ is at least $k$-$\SNM$-$(k-1)/k$ when $k+1$ is a power of $2$. 

\end{proof}

\end{document}